\documentclass[amsmath,amssymb]{article}

\usepackage{color}

\usepackage{graphicx}
\usepackage{amsmath}

\newtheorem{theorem}{Theorem}

\newtheorem{corollary}[theorem]{Corollary}

\newenvironment{proof}[1][Proof]{\textbf{#1.} }{\ \rule{0.5em}{0.5em}}

\usepackage{amsfonts}
\usepackage{mathrsfs}

\newcommand{\tensor}{\otimes}

\usepackage{makeidx}

\makeindex

\begin{document}
 
  \title{Quantum Filtering  for  Systems Driven by    {F}ermion Fields}

 \author{John E.~Gough\thanks{Institute for Mathematics and Physics, Aberystwyth University, SY23 3BZ, Wales, United Kingdom. Email: jug@aber.ac.uk }\and
 Madalin~I.~Guta\thanks{School of Mathematical Sciences, University of Nottingham, Nottingham, NG7 2RD, United Kingdom. Email: madalin.guta@nottingham.ac.uk}
\and Matthew R.~James\thanks{School of Engineering, Australian National University, Canberra, ACT 0200, Australia. Email: Matthew.James@anu.edu.au}\and
Hendra I.~Nurdin\thanks{School of Engineering, Australian National University, Canberra, ACT 0200, Australia. Email: Hendra.Nurdin@anu.edu.au}}

 \date{17 November,  2010}

\maketitle

\begin{abstract}
Recent developments in quantum technology mean that is it now possible to manipulate systems and measure fermion fields (e.g. reservoirs of electrons) at the quantum level. This progress has motivated some recent work on filtering theory for quantum systems driven by fermion fields by  Korotkov, Milburn and others. The purpose of this paper is to develop fermion filtering theory using the fermion quantum stochastic calculus. We explain that this approach has close connections to the classical filtering theory that is a fundamental part of the systems and control theory that has developed over the past 50 years.
\end{abstract}


\section{Introduction}
\label{sec:intro}

  A basic problem 
 in control and communication systems  is that of extracting information  from a signal that may contain noise. Problems of this kind are known as {\em filtering} problems. One common scenario concerns the problem of estimating variables of a system from partial, and typically noisy, information. Here, the word `system' refers to the entity of interest, which may be a machine being controlled, or it may be a signal model. As remarked in \cite[sec. 1.2]{AM79}, filtering theory developed in response to demands from applications. For example, the Kalman filter \cite{REK60} was developed at a time when significant efforts were underway in aerospace engineering in the early 1960's. The Kalman filter is the solution of a filtering problem based on a statistical model involving Gaussian stochastic processes and linear dynamics (statistical filtering dates back to Kolmogorov \cite{AK41} and Wiener \cite{NW49}). A more general theory of nonlinear filtering was  developed during the 1960's: 
 Kushner \cite{HK64},  Stratonovich  \cite{RS60},
  Duncan \cite{TD67}, Mortensen \cite{RM66}, and 
Zakai \cite{MZ69}. These filtering results boil down to determination of conditional expectations in dynamical contexts, and this  may be achieved using powerful tools from stochastic calculus,  including idealized Wiener process models for noise and It\={o} stochastic  differential equations,  \cite[Chapter 18]{RE82}.

  At the present time, developments in {\em quantum technology} are demanding methods for statistical estimation (among other things). Quantum technologies are those technologies that depend on the laws of quantum mechanics for their operation. Examples  of quantum technologies currently under development include quantum computers and atom lasers. Progress in quantum filtering has been slow, due  to both conceptual and practical experimental issues concerning the measurement of quantum systems. However, significant advances were made in the  field of quantum optics, a field of study  concerned with quantum properties of light and the interaction of light and matter. In quantum optics, thanks to the invention of the laser and other experimental developments, quantum effects became  accessible and a sophisticated theory of open systems, measurement theory, and quantum noise emerged. In particular, we mention the quantum filtering results due to  Belavkin \cite{VPB92}, \cite{VPB92a} that extend the thinking behind the above-mentioned classical filtering theory  (Belavkin was a student of Stratonovich in the 1960's). Belavkin employed quantum stochastic methods that were developed in the 1980's to describe quantum noise and a quantum generalization of the It\={o} calculus, \cite{HP84}, \cite{GC85}. Belavkin's results involved  an implementation of conditional expectation in a dynamical  quantum mechanical  context, \cite{BHJ07}.  We also mention related work in quantum optics by Carmichael \cite{HC93} and Wiseman and Milburn \cite{WM93}, \cite{WM93b} which used more direct methods, and employed the terms `stochastic master equation' and `quantum trajectories' in connection with quantum filtering.
  
\begin{figure}[h]
\begin{center}
\includegraphics{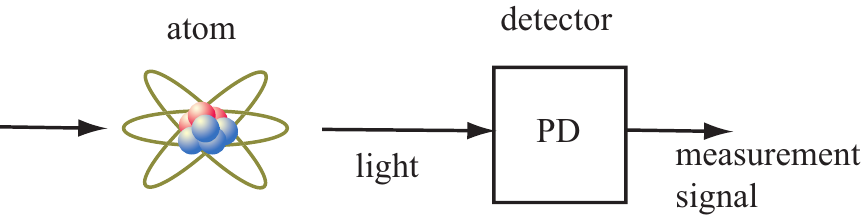} 
\end{center}
\caption{Schematic representation of the  detection of  light  emitted by an atom. The measurement signal  $Y(t)$ (integrated photocurrent) is proportional to the number of photons in the light incident on the detector up to time $t$.}
\label{fig:photo-detector-0}
\end{figure}

The absorption and emission of light by atoms was one of the earliest problems studied in quantum optics, 
\cite{AE17}, \cite[Chapter 1]{RL00}, \cite[Chapter 1]{GZ00}. 
Light is a type of electromagnetic field, which in quantum mechanics is described as a {\em boson} field, with  characteristic 
 Bose-Einstein statistics.  Photons are the well-known `particles' or `quanta' of light - an example of a {\em boson}.   Boson fields may be used in other contexts, such as in the description of vibrations in solid materials, or in the modeling of dissipation to a heat bath.
To give some indication of how quantum filtering may be used in quantum optics, consider the setup of Figure \ref{fig:photo-detector-0}, which illustrates the detection of light emitted by an atom. If we model the atom as a two-level system with ground and excited states, then the occupation number operator $n$ for the excited state is a physical observable of interest (a self-adjoint operator  with eigenvalues 0 and 1). The event corresponding to the eigenvalue 1 means that the atom is in the excited state and so contains one \lq{quanta}\rq \   of energy. Any energy in the atom may be emitted to the ambient field,  a dissipative process. Indeed, the mean value $\bar n(t) = e^{-\gamma t} \bar n(0)$ decays exponentially ($\gamma$ is a parameter describing the strength of the coupling of the atom to the light field). The photodetector (PD) shown in Figure \ref{fig:photo-detector-0} is taken to be an idealized device that produces a classical (i.e. non-quantum) photocurrent that is proportional to the number of photons in the field and 
may be processed using conventional analog or digital electronics. A quantum filter is such a processing system that is designed to provide estimates of, in this case, atomic observables (which are not directly accessible). 
The quantum filter for the conditional expectation $\hat n(t)$ of the occupation number  given the  photocurrent is 
  \begin{equation}
d \hat n(t) = -\gamma \hat n(t) dt - \hat n(t) (dY(t) - \gamma \hat n(t) dt) ,
\label{eq:boson-filter-n}
\end{equation}
where $Y(t)$ is the integrated photocurrent signal.
 This filter has a form that is familiar to control engineers, which combines a prediction term with an update term involving the innovations process  $W(t)$ defined by $dW(t)=dY(t) - \gamma \hat n(t) dt$.
We refer the reader to the book \cite{WM10} and the tutorial paper \cite{BHJ07} for further information on quantum filtering involving boson fields. 

In quantum field theory, there is another type of field  that is distinct from boson fields. These are called {\em fermion fields}, and the quanta of these fields are called {\em fermions}, the electron being an important example.  Fermion fields have Fermi-Dirac statistics. 
 Fermion fields arise  in mesoscopic systems, such as quantum dots, which are of considerable technological importance (e.g. for use in quantum computers).
Quantum dots may be fabricated in semiconductor materials to confine one or a few electrons to a region the size of a few nanometers. 
Fermion fields may be used to describe the flow of electrons at the quantum level. Figure \ref{fig:quantum-dot-1} provides a schematic representation of a quantum dot connected to two fermion field channels representing a source and a sink. 
Experimentally it has been much harder to extract quantum properties of mesoscopic systems compared to quantum optical systems, however,  it is now experimentally feasible to do so.

\begin{figure}[h]
\begin{center}
\includegraphics{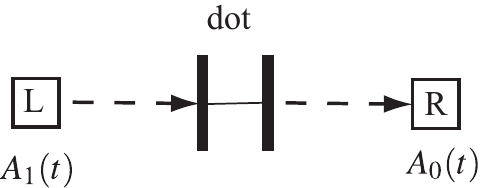} 
\end{center}
\caption{Schematic representation of quantum dot, through which current tunnels from the source ($L$) ohmic contact to the drain ($R$), \cite{GM00}. }
\label{fig:quantum-dot-1}
\end{figure}

Some important results have been obtained concerning  quantum filtering theory for the case of fermion fields.
Korotokov \cite{AK99}, \cite{AK01} developed a phenomenological approach using Bayesian methods, and Goan {\em et al} \cite{GM01}, \cite{GMWS01} adapted methods from quantum optics. 
The purpose of the present paper is to develop quantum filtering theory for systems driven by fermion fields using the fermion quantum stochastic calculus, 
Applebaum and Hudson, 
\cite{AH84}, Hudson and Parthasarathy \cite{HP86}, Milburn \cite{GM00}, Gardiner \cite{CWG04}. 
We explain  that quantum filtering, involving both boson and fermion fields, and classical filtering share a common mathematical foundation in conditional expectations, and that stochastic calculus provides powerful tools with which quantum and classical filtering problems can be solved. The key difference between the classical and quantum problems is that in the quantum  cases, the random variables and stochastic processes (such as quantum observables and field operators) have  non-commutative algebraic structures that are fundamental to quantum mechanics (classical variables are  represented as scalar valued functions and so commute under pointwise multiplication).
The principle algebraic distinctions between the boson and fermion cases are the commutation and anticommutation relations satisfied by these fields (which underly the Bose-Einstein and Fermi-Dirac statistics), and the need to take parity into account in the fermion case.

In an effort to be as concrete as possible, in this paper we develop the filtering theory using a model with just enough generality to solve filtering problems for two examples. The first example we consider (see Section \ref{sec:eq-dot}) concerns the quantum dot shown in Figure \ref{fig:quantum-dot-1}, where the electron flow in the right  channel is monitored, and this information is used to estimate the number of electrons (0 or 1) in the dot. Our second example is based on a more detailed model for the process of photodetection shown in Figure \ref{fig:photo-detector-0}, \cite[sec. 8.5]{GZ00}. This more detailed model is shown in Figure \ref{fig:photo-detector-1}, and involves one boson and two fermion field channels.

\begin{figure}[h]
\begin{center}
\includegraphics{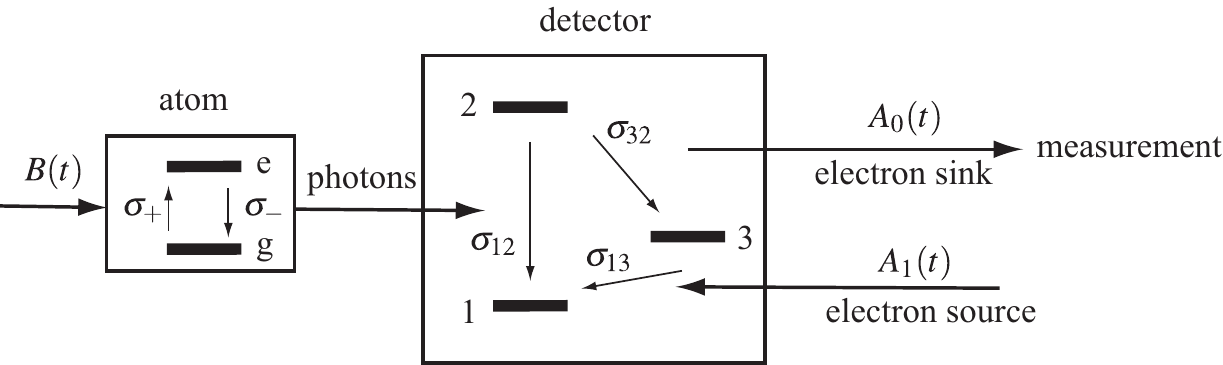} 
\end{center}
\caption{Schematic representation of a model for the  detection of the photons  emitted by an atom.  The model includes one boson field channel $B(t)$, two fermion field channels $A_0(t)$ and $A_1(t)$, a two-level system to describe the atom, and a three-level system capturing the essential behavior of the detector. The notation is defined in Section \ref{sec:detect}.}
\label{fig:photo-detector-1}
\end{figure}

The paper is organized as follows. In Section \ref{sec:classical} we review some basic ideas about classical filtering theory, and in particular, we summarize how the fundamental nonlinear filtering results may be obtained using classical stochastic calculus.  Section \ref{sec:model}  provides a description of the quantum stochastic model that is needed to formulate the filtering problem. This section include a brief review of some aspects of quantum mechanics, with an emphasis on describing the boson and fermion commutation relations. The main fermion filtering results are presented in Section \ref{sec:filter}, and the examples are given in Section \ref{sec:examples}. Two appendices briefly summarize some basic aspects of classical and quantum stochastic calculus, and show how   parity arises in systems with fermionic degrees of freedom.

{\em Notation:}
In what follows the symbols 
 $\mathbf{E}$  and $\mathbb{E}$ represent classical and quantum expectations, respectively. The commutator of two operators $A$ and $B$ is denoted
  $[A,B]=AB-BA$, while the anticommutator is written as $\{ A, B \} = AB+BA$.  The Dirac notation for a vector $\psi$ in a Hilbert space $\mathfrak{H}$ is $\vert \psi \rangle$, and   the inner product is written as $\langle \phi, \psi \rangle$ or $\langle \phi \vert  \psi \rangle$. The quantum expectation for an observable $X$ when the system is in a state described by the vector $\psi$ is $\mathbb{E}_{\psi}[X] = \langle \psi, X\psi \rangle$ or $\langle \psi \vert  X \vert \psi \rangle$. The adjoint of an operator $X$ is denoted by $X^\ast$.

\section{Classical Filtering Theory}
\label{sec:classical}

 In this section we review some of the fundamental concepts and results concerning classical  (non-quantum) filtering theory that will assist with understanding the quantum filtering results to be presented below.
 Let $\xi$ and $Y$ be classical random variables, with joint density $p_{\xi,Y}(x,y)$. In the absence of any measurement data, one may simply use the marginal density $p_\xi(x) = \int p_{\xi,Y}(x,y) dy$ to make inferences about $\xi$; for instance, the mean $\bar \xi = \int p_\xi(x)dx$ gives us an indication of the  value  an observation of $\xi$ may yield. If a value $y$ of $Y$ is observed, then the density for $\xi$ is revised to the conditional density
\begin{equation}
p_{\xi \vert Y} (x \vert y) = \frac{ p_{\xi,Y}(x,y) }{p_Y(y)} ,
\end{equation}
reflecting an increase in knowledge.   The conditional mean is defined to be
\begin{equation}
\hat \xi = \int x p_{\xi \vert Y} (x \vert y)  dx ,
\end{equation}
which we note is a function of the observed data $y$. More generally, if $\phi$ is an arbitrary function, we may compute the conditional expected value
\begin{equation}
\hat \phi =  \int \phi(x) p_{X \vert Y} (x \vert y)  dx
\label{eq:hat-phi-c}
\end{equation}
of the random variable $\phi(\xi)$. The RHS of (\ref{eq:hat-phi-c}) is an explicit expression for the conditional expectation, which is denoted more generally as
\begin{equation}
\hat \phi = \pi(\phi)= \mathbf{E}[ \phi(\xi) \vert Y].
\end{equation}

As a simple example, suppose that
\begin{equation}
Y=\xi+V,
\label{eq:Y-basic}
\end{equation}
where $\xi$ and $V$ are independent Gaussian random variables with   means $\bar \xi$ and $0$, and variances $\Sigma_\xi$ and $1$. The expression (\ref{eq:Y-basic}) for the observations $Y$ is an instance of the fundamentally important \lq\lq{signal plus noise}\rq\rq \ models widely employed in control and communications systems.
The conditional mean $\hat \xi$ is also Gaussian \cite{AM79} and is given by
\begin{equation}
\hat \xi = \bar \xi + \Sigma_\xi( 1+\Sigma_\xi)^{-1}( Y-\bar \xi) .
\label{eq:X-cond-mean-gaussian-c}
\end{equation}
This expression shows how the conditional mean updates the prior mean $\bar \xi$ by the addition of a term $W=Y-\bar \xi$, called the innovation. The innovation represents the new information about $\xi$ that is gained from an observation of $Y$.

The innovation is related to the minimum variance property of the conditional mean, which means that $\hat \xi$ minimizes the variance of the error   \lq\lq{error}\rq\rq \ $E=\tilde \xi-\hat \xi$ over all estimators $\tilde \xi$. Here, estimators $\tilde \xi$ are random variables that are functions of the observation $Y$, that is, random variables that belong to the subspace $\mathscr{Y}$ generated by $Y$.\footnote{That is, $\mathscr{Y}$ is the subspace of square integrable random variables that are measurable with respect to the $\sigma$-algebra $\sigma(Y)$ generated by $Y$.}
 The least squares property has a nice geometrical interpretation,  where $\hat \xi$ is the orthogonal projection of $\xi$ onto $\mathscr{Y}$, Figure \ref{fig:c-exp-1}.

\begin{figure}[h]
\begin{center}
\includegraphics{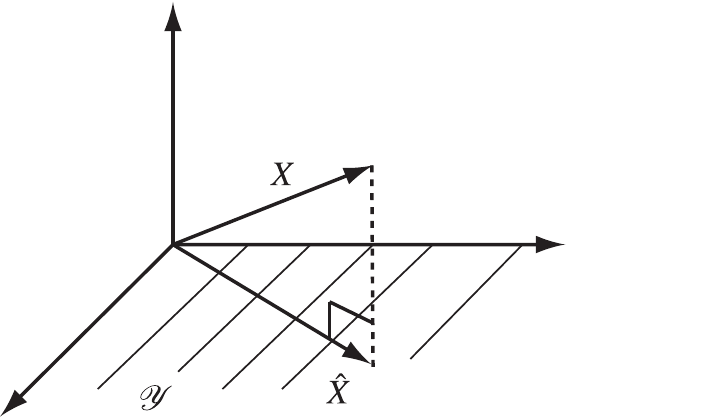}
\caption{The conditional expectation $\hat \xi=\mathbf{E}[ \xi \vert  Y]$ is the orthogonal projection of $X$ onto the subspace $\mathscr{Y}$, \cite[Fig. 5.2-1]{AM79}.
}
\label{fig:c-exp-1}
\end{center}
\end{figure}

In general, if $\phi$ and $Y$ have well defined expectance then the conditional expectation $\mathbf{E}[ \phi \vert Y ]$  
is {\em defined} to be the unique (up to a set of measure zero with respect to the underlying probability measure) random variable  $\hat\phi = \mathbf{E}[ \phi \vert Y] \in \mathscr{Y}$ such that \cite[Chapter 1]{WH85}
\begin{equation}
 \mathbf{E}[ \tilde \phi  \phi  ]  = \mathbf{E}[  \tilde \phi \hat \phi ]
\label{eq:qm-prob-cexp-1}
\end{equation}
 for all bounded  random variables $\tilde \phi \in \mathscr{Y}$.

\subsection{Nonlinear Filtering of Classical Systems in Continuous Time}

In control and communications systems, signals are often modeled as
stochastic processes, which are sequences of random variables. For linear
Gaussian systems, the Kalman filter computes the conditional mean and
covariances in a causal manner as time progresses, as described in 
the book \cite{AM79} for discrete time
systems.
Of interest to us here are general approaches to filtering in continuous
time that exploit the power of the stochastic calculus, \cite[Chapter 18]{RE82}. We refer the reader to Appendix \ref{sec:app-stoch-classical} for some basic concepts concerning stochastic integrals and the It\={o} rule.

Consider the stochastic system expressed as a system of It\={o}
stochastic differential equations 
\begin{eqnarray}
d \xi(t) &=& g( \xi(t)) dt + dV_1(t)  \label{eq:c-X-dyn} \\
dY(t) &=& h(\xi(t)) dt + dV_2(t)  \label{eq:c-Y-dyn}
\end{eqnarray}
where $V_1(t)$ and $V_2(t)$ are independent Wiener processes. Here, $\xi(t)$
represents a signal of interest that is not directly accessible to
observation. Instead, a signal $Y(t)$ is observed.

Let us first consider the dynamics of expected  values for this system. For any sufficiently regular
function $\phi$ write $j_t(\phi)=\phi(\xi(t))$ and define 
\begin{equation}
\mu_t( \phi) = \mathbf{E}[ j_t(\phi) ]
\end{equation}
for the mean of the random variable $j_t(\phi)$. Now by the It\={o} rule we have 
\begin{equation}
d j_t(\phi) = j_t( \frac{d \phi}{dx} g ) dt + j_t( \frac{d \phi}{dx}
)dV_1(t) ) + j_t( \frac{1}{2} \frac{d^2 \phi}{dx^2} )dt,
\end{equation}
and so 
\begin{equation}
\frac{d}{dt} \mu_t( \phi) = \mu_t ( \mathcal{L}(\phi) ), 
 \label{eq:X-mean-c}
\end{equation}
where 
\begin{equation}
\mathcal{L}(\phi) = \frac{1}{2} \frac{d^2 \phi}{dx^2} + \frac{d \phi}{dx} g
\end{equation}
is the generator of the Markov process $\xi(t)$ given by (\ref{eq:c-X-dyn}).
The  Kolmogorov equation for the density $p(x,t)$ defined by 
$\mu_t(\phi) = \int \phi(x) p(x,t)dx$ is
\begin{equation}
\frac{\partial }{\partial t} p = \mathcal{L}^\ast( p) ,
\end{equation}
where $\mathcal{L}^\ast(p)=\frac{1}{2} \frac{d^2}{dx^2}p -  \frac{d }{dx} (gp)$. 

Now suppose we wish to determine the differential equation for the
conditional expectation 
\begin{equation}
\pi_t( \phi) = \mathbb{E}[ j_t(\phi) \vert \mathscr{Y}_t ] ,
\end{equation}
where $\mathscr{Y}_t$ is generated by the observations $Y(s)$, $0 \leq s
\leq t$. For convenience, we write $h_t$ for $h(\xi (t))$. We might expect
the filter equation to be a modification of the mean equation (\ref
{eq:X-mean-c}) along the lines of (\ref{eq:X-cond-mean-gaussian-c}). There
are several commonly used methods for finding the filter equations,
including the martingale approach (which lead to the Kushner-Stratonovich
equation \cite{HK64}, \cite{RS60}), the reference probability method (giving
the Duncan-Mortensen-Zakai equation \cite{TD67}, \cite{M68}, \cite{MZ69}), and
the characteristic function method, \cite{HSM05},  \cite{VPB93}, \cite{VPB92}
which we will use in this paper. 

We suppose that the filter has the form 
\begin{equation}
d \pi_t(\phi) = F_t(\phi)dt + G_t( \phi) dY(t)
\end{equation}
and consider, for any square-integrable function $f$, the stochastic process 
$C_t (f)$ defined to be the solution of 
\begin{equation}
d C_t (f) = f(t) C_t (f) dY(t), \ \ C_f(0)=I.
\end{equation}
By the definition of conditional expectation (recall (\ref{eq:qm-prob-cexp-1}%
)), we have 
\begin{equation}
\mathbb{E}[ j_t(\phi) C_t (f) ] = \mathbb{E}[ \pi_t(\phi) C_t (f)]
\label{eq:c-char-1}
\end{equation}
for all $f$. We note that the It\={o} product rule implies $I+II+III=0$ where

\begin{eqnarray*}
I &=&\mathbb{E}\left[ \left( d\pi _{t}\left( \phi \right) -dj_{t}\left( \phi
\right) \right) C_t (f) \right] , \\
II &=&\mathbb{E}\left[ \left( \pi _{t}\left( \phi \right) -j_{t}\left( \phi
\right) \right) dC_t (f) \right] , \\
III &=&\mathbb{E}\left[ \left( d\pi _{t}\left( \phi \right) -dj_{t}\left(
\phi \right) \right) dC_t (f) \right] .
\end{eqnarray*}

For the model above, we then have 
\begin{eqnarray*}
I &=&\mathbb{E}\left[ \left\{ F_{t}\left( \phi \right) +h_{t}G_{t}\left(
\phi \right) -j_{t}\left( \mathcal{L}\phi \right) \right\} C_{t}\left(
f\right) \right] dt\equiv \mathbb{E}\left[ \left\{ F_{t}\left( \phi \right)
+\pi _{t}(h_{t})G_{t}\left( \phi \right) -\pi _{t}\left( \mathcal{L}\phi
\right) \right\} C_{t}\left( f\right) \right] dt, \\
II &=&\mathbb{E}\left[ \left\{ \pi _{t}\left( \phi \right) -j_{t}\left( \phi
\right) \right\} f\left( t\right) C_{t}\left( f\right) h_{t}\right] dt\equiv
f\left( t\right) \mathbb{E}\left[ \left\{ \pi _{t}\left( \phi \right) \pi
_{t}\left( h_{t}\right) -j_{t}\left( \phi h_{t}\right) \right\} C_{t}\left(
f\right) \right] dt, \\
III &=&\mathbb{E}\left[ G_{t}\left( \phi \right) f\left( t\right)
C_{t}\left( f\right) \right] dt.
\end{eqnarray*}

As $f(t)$ was arbitrary, we may separate the $f$ independent and dependent
terms to obtain 
\begin{equation*}
F_{t}\left( \phi \right) =\pi _{t}\left( \mathcal{L} \phi \right) -\pi
_{t}\left( h_{t}\right) G_{t},\quad G_{t}\left( \phi \right) =\pi _{t}\left(
\phi h_{t}\right) -\pi _{t}\left( \phi \right) \pi _{t}\left( h_{t}\right)
\end{equation*}
so that 
\begin{equation*}
d\pi _{t}\left( \phi \right) =\pi _{t}\left( \mathcal{L} \phi \right)
dt+\left\{ \pi _{t}\left( \phi h_{t}\right) -\pi _{t}\left( \phi \right) \pi
_{t}\left( h_{t}\right) \right\} dW_{t}
\end{equation*}
where the innovations process is a $\mathscr{Y}_t$ martingale (actually a
standard Wiener process) given by 
\begin{equation}
dW(t) = dY(t) - \pi_t(h_t) dt .
\end{equation}

The conditional density $\hat p(x,t)$ defined by $\int \phi(x) \hat p(x,t)dx = \mathbf{E}[ j_t(\phi) \vert \mathscr{Y}_t ]$ satisfies the equation
\begin{equation}
d \hat  p = \mathcal{L}^\ast(\hat p) dt +  (h - \pi_t(h) )\hat p dW(t).
\end{equation}

In the special case of linear systems, say $g(\xi)=a \xi$ and $h(\xi)=c \xi$, with initial Gaussian states, the  process $\xi(t)$ is Gaussian, with
mean $\bar \xi(t) = \mathbf{E}[ \xi(t)]$ satisfying the equation
\begin{equation}
\dot{\bar{\xi}}(t) = a \bar \xi(t),
\end{equation}
and variance $\Gamma(t) = \mathbf{E}[ ( \bar \xi(t) - \xi(t) )^2 ]$ satisfying
\begin{equation}
\dot \Gamma(t) = 2a \Gamma(t) + 1.
\end{equation}
 The
conditional mean $\hat \xi(t)=\mathbf{E}[\xi(t) \vert \mathscr{Y}_t]$ is Gaussian and is given by the Kalman filter equations
\begin{eqnarray}
d \hat \xi (t) &=& a \hat \xi(t) dt + c \Sigma(t) ( dY(t) - c \hat \xi(t) dt)
\label{eq:kalman-1}
\\
\dot \Sigma(t) &=& 2a \Sigma(t) + 1 + c^2 \Sigma^2.
\label{eq:kalman-2}
\end{eqnarray}
Here, $\Sigma(t) = \mathbf{E}[ (\hat \xi(t) - \xi(t) )^2 \vert \mathscr{Y}_t]$ is the conditional variance, a deterministic quantity (a special feature of the linear-Gaussian case).

\section{Quantum Stochastic Model}

\label{sec:model}

In this paper we are interested in a quantum system interacting with quantum
fields, for instance as sketched in Figure \ref{fig:fermion-model1}. Here
the ``box" represents a quantum system with finitely many degrees of
freedom, such as an atom or a  quantum dot or a photodetector. The 
input/output lines are
quantum fields,  representing reservoirs of electrons or photons coupled
to the system. The model for this system has a natural input-output
structure, 
with an  input being the incident part of the field, while  the output is the reflected
part of the field which 
 carries  away
information about the system. Our main goal in this paper concerns
estimation of system variables given the results of monitoring the output of
fermion channel $0$. The purpose of this section is to describe a quantum
mechanical model for this system.

\begin{figure}[h]
\begin{center}
\includegraphics{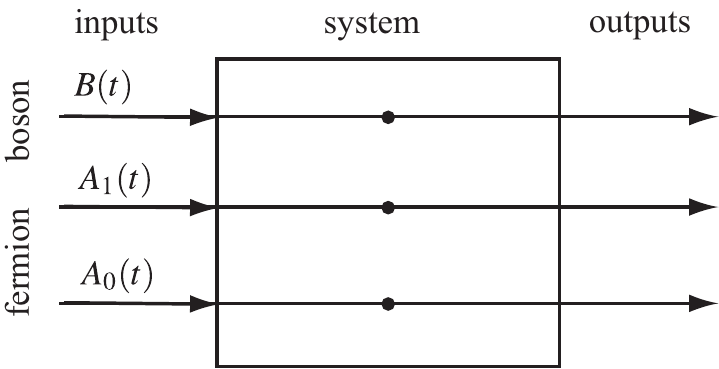} 
\end{center}
\caption{Schematic representation of a system coupled to boson $B$ and
fermion $A_0$, $A_1$ fields.}
\label{fig:fermion-model1}
\end{figure}

\subsection{Quantum Mechanics}
\label{sec:model-qm}

Quantum mechanics was developed in the 20th century in order to explain the behavior of light and matter on a small scale, \cite{EM98}, \cite{AFP09}.    
Central to quantum mechanics are the notions of observables $X$,
which are mathematical representations of physical quantities that
can (in principle) be measured, and state\footnote{
The word \lq{state}\rq \  is heavily overloaded in the physical and engineering sciences, though its   core meaning as a way of minimally storing dynamical and statistical information is common to all interpretations.
}
 vectors $\psi$, which
summarize the status of physical systems and permit the
calculation of expected values  of observables. State vectors may be
described mathematically as elements of a Hilbert space
$\mathfrak{H}$, while observables are self-adjoint operators on
$\mathfrak{H}$.  The expected value of an observable $X$ when in
state $\psi$ is given by the inner product  $\langle \psi, X \psi
\rangle$.

The simplest non-trivial quantum system has two energy levels and is often
used to model ground and excited states of atoms. Since the advent
of quantum computing, this system is also known as the qubit, the
unit of quantum information. 
The Hilbert space
for this system is $\mathfrak{H}=\mathbf{C}^2$, the two-dimensional
complex vector space. The space of all operators is
 spanned by the Pauli matrices \cite[sec.
2.1.3]{NC00}, \cite[sec. 9.1.1]{GZ00}:
\begin{eqnarray*}
\sigma_0 = I= \left(  \begin{array}{cc}  1 & 0 \\ 0 & 1\end{array}  \right), \ \
\sigma_x = I= \left(  \begin{array}{cc}  0 & 1 \\ 1 & 0 \end{array}  \right), \ \
\sigma_y = I= \left(  \begin{array}{cc}  0 & -i \\ i & 0 \end{array}  \right),  
\sigma_z = I= \left(  \begin{array}{cc}  1 & 0 \\ 0 & - 1\end{array}  \right).
\end{eqnarray*}
Since in general operators do not commute, the {\em commutator} 
\begin{equation}
[A,B]= AB-BA,
\end{equation}
which is a measure of the failure to commute, is frequently used. 
The basic
{\em commutation relations} for the Pauli matrices are $[\sigma_x, \sigma_y]=2i \sigma_z$,
$[\sigma_y, \sigma_z]=2i \sigma_x$, and $[\sigma_z, \sigma_x]=2i
\sigma_y$.

The energy levels correspond to the eigenvalues $\pm 1$ of
$\sigma_z$, and the corresponding eigenstates are referred to as the ground $\vert - 1 \rangle$ and excited states $\vert   1 \rangle$. These eigenstates form a basis for $\mathfrak{H}$. In quantum mechanics, state vectors $\psi$ are normalized to one, so that we may write
\begin{equation}
\psi = \alpha \vert -1 \rangle + \beta \vert 1 \rangle ,
\end{equation}
where $\vert \alpha \vert^2+\vert \beta \vert^2=1$.
The operators $\sigma_{\pm}=\frac{1}{2}(\sigma_x \pm i\sigma_y)$ are known as
raising and lowering operators, and have
  actions  $\sigma_+ \vert -1 \rangle = \vert 1 \rangle$ and $\sigma_- \vert  1 \rangle = \vert -1 \rangle$, as
  illustrated in 
Figure \ref{fig:levels} (a). Thus $\sigma_+$ corresponds to the creation of a quanta of energy in the system, while $\sigma_-$ corresponds to destruction. The number operator $n=\sigma_+ \sigma_-$ counts the number of quanta in the system, which is in this case is 0 or 1 (ground or excited).
Note that $\sigma_+$ is the adjoint of $\sigma_-$: $\sigma_+=\sigma_-^\ast$.
The raising and lowering operators satisfy the {\em anticommutation relation} 
\begin{equation}
\{ \sigma_-, \sigma_-^\ast  \}=1,
\label{eq:anti-ccr}
\end{equation}
where the {\em anticommutator} is defined by
\begin{equation}
\{ A,B \} = AB + BA.
\end{equation}
Note also that $\sigma_-^2=0$.

The postulates of quantum mechanics state that if an observable $A$ is measured, the allowable measurement values are the eigenvalues $\lambda_j$ of $A$. If the system is in state $\psi$, the probability of observing the outcome $\lambda_j$ is
\begin{equation}
\mathrm{Prob}(\lambda_j) = \langle \psi, P_j \psi \rangle,
\end{equation}
where $P_j$ is the projection associated with the eigenvalue $\lambda_j$. When an eigenvalue $\lambda_j$ is recorded, the state \lq\lq{collapses}\rq\rq \ to
\begin{equation}
\psi' = \frac{P_j \psi}{\sqrt{\mathrm{Prob}(a_j)}} .
\end{equation}

This brief discussion of quantum measurement theory has focused on ideal measurements when the system is in a pure state (a state vector). Density operators $\rho$ (self-adjoint non-negative operators  of trace one) provide a more general notion of state, where the probability of outcome is given by $\mathrm{Prob}(\lambda_j) = \mathrm{tr}[ \rho P_j]$, and the collapsed state is $\rho' = \frac{P_j \rho P_j}{ \mathrm{Prob}(\lambda_j) }$.
A theory of generalized measurements was developed largely during the 70's an 80's that allows for non-ideal circumstances \cite[Chapters 2 and 8]{NC00}, and by the addition of  ancilla systems all of quantum measurement theory can be seen to be a consequence of quantum conditional expectation, \cite{BHJ07}.

Another fundamental postulate of quantum mechanics states that the dynamics of states and observables is unitary. That is,
 state vectors and observables evolve
according to
\begin{equation}
\psi(t) = U(t) \psi ,  \ \ \ X(t) = U^\ast(t) X U(t) ,
\label{eq:pictures}
\end{equation}
where the operator $U(t)$
is unitary ($U^\ast(t) U(t)= U(t) U^\ast(t) = I$).
The unitary $U(t)$ is the solution of the {\em Schr\"{o}dinger equation}
\begin{equation}
 \frac{d}{dt} U(t) = -i H U(t) ,
\end{equation}
where the observable $H$ is called the {\em Hamiltonian}. 
The  expressions in (\ref{eq:pictures})   provide two equivalent descriptions (dual), the
former is referred to as the {\em Schr\"{o}dinger picture}, while
the latter is the {\em Heisenberg picture}.  
The differential equations in the Schr\"{o}dinger and Heisenberg pictures are respectively
\begin{equation}
\frac{d}{dt} \psi(t) = -i H \psi(t) 
\end{equation}
and
\begin{equation}
\frac{d}{dt} X(t) = - i [ X(t), H(t) ] ,
\label{prelim:hp1}
\end{equation}
where $H(t)=U^\ast(t) H U(t)$.

\begin{figure}[h]
\begin{center}
\setlength{\unitlength}{1879sp}
\begingroup\makeatletter\ifx\SetFigFont\undefined%
\gdef\SetFigFont#1#2#3#4#5{%
  \reset@font\fontsize{#1}{#2pt}%
  \fontfamily{#3}\fontseries{#4}\fontshape{#5}%
  \selectfont}%
\fi\endgroup%
\begin{picture}(9255,5270)(1036,-10109)
\put(10201,-5836){\makebox(0,0)[lb]{\smash{{\SetFigFont{5}{6.0}{\familydefault}{\mddefault}{\updefault}{\color[rgb]{0,0,0}$n=3$}%
}}}}
\thicklines
{\color[rgb]{0,0,0}\put(2101,-8161){\line( 1, 0){1800}}
}%
{\color[rgb]{0,0,0}\put(7801,-9361){\line( 1, 0){1800}}
}%
{\color[rgb]{0,0,0}\put(1801,-8461){\vector( 0,-1){600}}
}%
{\color[rgb]{0,0,0}\put(7501,-8461){\vector( 0,-1){600}}
}%
{\color[rgb]{0,0,0}\put(7801,-8161){\line( 1, 0){1800}}
}%
{\color[rgb]{0,0,0}\put(7801,-6961){\line( 1, 0){1800}}
}%
{\color[rgb]{0,0,0}\put(7801,-5761){\line( 1, 0){1800}}
}%
{\color[rgb]{0,0,0}\put(7501,-7261){\vector( 0,-1){600}}
}%
{\color[rgb]{0,0,0}\put(7501,-5986){\vector( 0,-1){600}}
}%
{\color[rgb]{0,0,0}\put(9901,-9061){\vector( 0, 1){600}}
}%
{\color[rgb]{0,0,0}\put(9901,-7861){\vector( 0, 1){600}}
}%
{\color[rgb]{0,0,0}\put(9901,-6661){\vector( 0, 1){600}}
}%
{\color[rgb]{0,0,0}\put(9901,-5461){\vector( 0, 1){600}}
}%
{\color[rgb]{0,0,0}\put(7501,-4861){\vector( 0,-1){600}}
}%
{\color[rgb]{0,0,0}\put(4201,-9061){\vector( 0, 1){600}}
}%
\put(1201,-8836){\makebox(0,0)[lb]{\smash{{\SetFigFont{5}{6.0}{\familydefault}{\mddefault}{\updefault}{\color[rgb]{0,0,0}$\sigma_-$}%
}}}}
\put(1051,-8161){\makebox(0,0)[lb]{\smash{{\SetFigFont{5}{6.0}{\familydefault}{\mddefault}{\updefault}{\color[rgb]{0,0,0}excited}%
}}}}
\put(1051,-9361){\makebox(0,0)[lb]{\smash{{\SetFigFont{5}{6.0}{\familydefault}{\mddefault}{\updefault}{\color[rgb]{0,0,0}ground}%
}}}}
\put(2776,-10036){\makebox(0,0)[lb]{\smash{{\SetFigFont{5}{6.0}{\familydefault}{\mddefault}{\updefault}{\color[rgb]{0,0,0}(a)}%
}}}}
\put(6826,-9436){\makebox(0,0)[lb]{\smash{{\SetFigFont{5}{6.0}{\familydefault}{\mddefault}{\updefault}{\color[rgb]{0,0,0}vacuum}%
}}}}
\put(6901,-6436){\makebox(0,0)[lb]{\smash{{\SetFigFont{5}{6.0}{\familydefault}{\mddefault}{\updefault}{\color[rgb]{0,0,0}$a$}%
}}}}
\put(10276,-8836){\makebox(0,0)[lb]{\smash{{\SetFigFont{5}{6.0}{\familydefault}{\mddefault}{\updefault}{\color[rgb]{0,0,0}$a^\ast$}%
}}}}
\put(10276,-7636){\makebox(0,0)[lb]{\smash{{\SetFigFont{5}{6.0}{\familydefault}{\mddefault}{\updefault}{\color[rgb]{0,0,0}$a^\ast$}%
}}}}
\put(10276,-6361){\makebox(0,0)[lb]{\smash{{\SetFigFont{5}{6.0}{\familydefault}{\mddefault}{\updefault}{\color[rgb]{0,0,0}$a^\ast$}%
}}}}
\put(6901,-5236){\makebox(0,0)[lb]{\smash{{\SetFigFont{5}{6.0}{\familydefault}{\mddefault}{\updefault}{\color[rgb]{0,0,0}$a$}%
}}}}
\put(10201,-5236){\makebox(0,0)[lb]{\smash{{\SetFigFont{5}{6.0}{\familydefault}{\mddefault}{\updefault}{\color[rgb]{0,0,0}$a^\ast$}%
}}}}
\put(6901,-7636){\makebox(0,0)[lb]{\smash{{\SetFigFont{5}{6.0}{\familydefault}{\mddefault}{\updefault}{\color[rgb]{0,0,0}$a$}%
}}}}
\put(6901,-8836){\makebox(0,0)[lb]{\smash{{\SetFigFont{5}{6.0}{\familydefault}{\mddefault}{\updefault}{\color[rgb]{0,0,0}$a$}%
}}}}
\put(4576,-8836){\makebox(0,0)[lb]{\smash{{\SetFigFont{5}{6.0}{\familydefault}{\mddefault}{\updefault}{\color[rgb]{0,0,0}$\sigma_+$}%
}}}}
\put(8476,-10036){\makebox(0,0)[lb]{\smash{{\SetFigFont{5}{6.0}{\familydefault}{\mddefault}{\updefault}{\color[rgb]{0,0,0}(b)}%
}}}}
\put(8701,-5011){\rotatebox{90.0}{\makebox(0,0)[lb]{\smash{{\SetFigFont{5}{6.0}{\rmdefault}{\mddefault}{\updefault}{\color[rgb]{0,0,0}...}%
}}}}}
\put(10276,-9436){\makebox(0,0)[lb]{\smash{{\SetFigFont{5}{6.0}{\familydefault}{\mddefault}{\updefault}{\color[rgb]{0,0,0}$n=0$}%
}}}}
\put(10201,-8236){\makebox(0,0)[lb]{\smash{{\SetFigFont{5}{6.0}{\familydefault}{\mddefault}{\updefault}{\color[rgb]{0,0,0}$n=1$}%
}}}}
\put(10201,-7036){\makebox(0,0)[lb]{\smash{{\SetFigFont{5}{6.0}{\familydefault}{\mddefault}{\updefault}{\color[rgb]{0,0,0}$n=2$}%
}}}}
{\color[rgb]{0,0,0}\put(2101,-9361){\line( 1, 0){1800}}
}%
\end{picture}%

\caption{Energy level diagrams. (a) Two-level atom (qbit). (b) Harmonic oscillator.}
\label{fig:levels}
\end{center}
\end{figure}
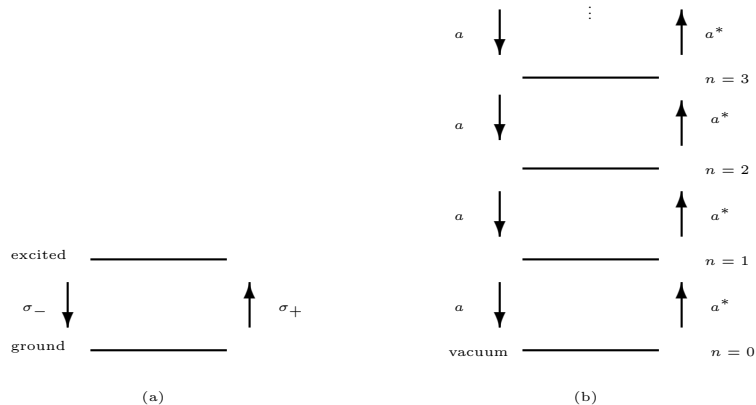

Another  basic example is that of a particle moving in a potential well,
\cite[Chapter 14]{EM98}. The position and momentum of the particle
are represented by   observables $Q$ and $P$, respectively,
defined by
\begin{equation}
(Q \psi)(q) = q \psi (q), \ \ \ (P\psi)(q) = -i    \frac{d}{dq} \psi(q)
\end{equation}
for $\psi \in \mathfrak{H} = L^2(\mathbf{R})$. Here,   $q\in \mathbf{R}$
represents  position values.  The position and momentum operators satisfy the
commutation relation $[Q,P]=i$. The dynamics of the particle
is  determined by the Hamiltonian
$
H = \frac{P^2}{2m} + \frac{1}{2}  m \omega^2 Q^2
$
 (here, $m$ is the mass of the particle,
and $\omega$ is the frequency of oscillation).

Energy eigenvectors $\psi_n$ are defined by the equation $H \psi_n
= E_n \psi_n$ for real numbers $E_n$. The system has a discrete
energy spectrum $E_n = (n+\frac{1}{2})  \omega$,
$n=0,1,2,\ldots$. The state $\psi_0$ corresponding to $E_0$ is
called the {\em ground state}. The {\em annihilation operator}
\begin{equation}
a = \sqrt{\frac{m \omega}{2} }(Q+ i \frac{P}{2m\omega})
\end{equation}
and the creation operator $a^\ast$ lower and raise energy levels,
respectively: $a \psi_n =  \sqrt{n} \psi_{n-1}$, and $a^\ast
\psi_n = \sqrt{n+1} \psi_{n+1}$, see Figure \ref{fig:levels} (b).
 They satisfy the canonical
commutation relation 
\begin{equation}
[a,a^\ast]=1.
\label{eq:ccr}
\end{equation}
 In terms of these operators,
the Hamiltonian  can be expressed as $H =   \omega( a^\ast a +
\frac{1}{2})$. Here, $n=a^\ast a$ is the number or counting operator, with eigenvalues  $0,1,2,\ldots$.
 Using (\ref{prelim:hp1}), the annihilation
operator evolves according to
$
\frac{d}{dt} a(t) = -i\omega a(t)
$
with solution $a(t)=e^{-i \omega t} a$.  Note that also $a^\ast(t)
= e^{i \omega t} a^\ast$, and so commutation relations are
preserved by the unitary dynamics:
$[a(t),a^\ast(t)]=[a,a^\ast]=1$. Because of the oscillatory nature
of the dynamics, this system is often referred to as the {\em
quantum harmonic oscillator}.

\subsection{Boson and Fermion Fields}
\label{sec:model-fields}

We may think of observables as \emph{quantum random variables}, and  
the key distinction with classical probability is that quantum random
variables do not in general commute. Indeed, if $(\Omega,\mathcal{F},P)$ is a classical probability space 
then classical bounded real-valued random variables in $L^{\infty}(\Omega,\mathcal{F},P)$ have an interpretation as 
multiplication operators that map the Hilbert space $L^2(\Omega,\mathcal{F},P)$ to itself. Since all such operators commute with one another,
bounded classical real-valued random variables are thus isomorphic to (and can be viewed as) commuting observables on $L^2(\Omega,\mathcal{F},P)$; see \cite{BHJ07} for further discussions, including the case of unbounded classical random variables. 
 
In this section we turn to the notion of \emph{quantum stochastic processes} which are used to provide tractable
models for how quantum systems evolve in dissipative environments, such as a
quantum dot in an electron field, or an atom in an electromagnetic field. A
detailed introduction to quantum fields is beyond the scope of the present
paper, and we refer interested readers to the references for more
information, particularly \cite[Chapter 21]{EM98}, \cite[Chapter II]{KRP92}, 
\cite[Section 4]{BHJ07}. Here we give some basic ideas needed for the
filtering results to follow.

In quantum field theory, a one dimensional 
quantum field (with parameter  $t$) consists of a collection of systems each with annihilation $a(t)$ and creation operators $a^\ast(t)$ used to describe the annihilation and
creation of quanta or particles at index location or point $t$. $a(t)$ and $a^\ast(t)$ are referred to
as {\em field operators}, the annihilation and creation field operators, respectively. The index $t$ may
represent a range of variables, including position, frequency and time, and we assume here that $t$ lies in a continuous interval $T$ in $\mathbb{R}$. 
Basic considerations lead to the postulate that the annihilation and
creation operators must satisfy either the commutation relations 
\begin{equation}
[ a(t), a^\ast(t') ]  = \delta(t-t'), 
 \label{eq:boson}
\end{equation}
or the anticommutation relations 
\begin{equation}
\{ a(t), a^\ast(t') \} = \delta(t-t'), 
 \label{eq:fermion}
\end{equation}
for all $t,t' \in T$, where $\delta(t)$ denotes the Dirac delta distribution.

Fields that satisfy the commutation relations (\ref{eq:boson}) are called 
\emph{boson} fields (e.g. photons), while fields that satisfy the
anticommutation relations (\ref{eq:fermion}) are called \emph{fermion}
fields (e.g. electrons). In this paper we will take the parameter  $t$ to be time
and  $T=[0,\infty)$. In this case $a(t)$ has the interpretation of annihilation of
a photon (in the case of a bosonic field) or electron (in the case of fermionic field) at time $t$, whereas
$a^*(t)$ has the interpretation of creation a photon (in the case of bosons) or electron (in the case of fermions) at time $t$.
One can imagine these fields as a continuous collection or stream of distinct quantum systems (one quantum system for each $t$) hence, informally, quantum fields can be defined on some continuous tensor product Hilbert space $\mathcal{H}= \otimes_{t \in [0,\infty)} \mathcal{H}_t$, where $\mathcal{H}_t$ is a Hilbert space for each $t$ (of the quantum system arriving at time $t$). Although such an object can be rigorously defined and constructed, from a mathematical viewpoint it is much easier not to work directly with the field operators $a(t)$ and $a^\ast(t)$ but with their integrated versions, the so-called smeared quantum field operators, as will be discussed below. Smeared quantum field operators can be constructed on Hilbert spaces known as Fock spaces (symmetric Fock space $\mathfrak{F}_{sym}$ for bosons and antisymmetric Fock space $\mathfrak{F}_{antisym}$ for Fermions) which have the character of  a continuous tensor product Hilbert space. Modulo the specification of the statistics of the field, a quantum field has the character of a quantum version of white noise, while its integrated version can be viewed as a quantum  independent increment process.  Thus, exploiting the properties of smeared quantum fields, Hudson and Parthasarathy \cite{HP84} were able to develop a quantum stochastic calculus which is essentially a quantum version of the It\={o} stochastic calculus. 

The model we use to describe the system shown in Figure \ref{fig:fermion-model1}  
employs boson  and fermion fields $b(t)$ and $a(t)$, respectively, parametrized by time $t \in [0,\infty)$
which accounts for the time evolution of fields interacting with the system
(e.g. an atom or quantum dot) at a fixed spatial location. In the remainder of this section we describe the quantum
stochastic calculus that has been developed to facilitate modeling and
calculations involving these fields, \cite{HP84}, \cite{AH84}, \cite{GC85}, 
\cite{HP86}, \cite{KRP92}, \cite{GZ00}, \cite{BHJ07}, \cite{WM10}. Some basic aspects of quantum stochastic integrals and the quantum It\={o} rule are discussed in Appendix \ref{sec:app-stoch-quantum}.

The  boson field channel  $B$ in Figure \ref{fig:fermion-model1}  
is defined on a symmetric Fock space $\mathfrak{F}_{sym}$. The commutation relations for the boson field are 
$[b(t), b^\ast(t^{\prime}) ] = \delta(t-t^{\prime})$, from  (\ref{eq:boson}).
For a boson channel in a Gaussian state, the following singular expectations
may be assumed: 
\begin{eqnarray}
\langle b^\ast(t) b(t^{\prime}) \rangle = N \delta(t-t^{\prime}), \ \
\langle b(t) b^\ast(t^{\prime}) \rangle = (N+1) \delta(t-t^{\prime}), \\
\langle b(t) b(t^{\prime}) \rangle = M \delta(t-t^{\prime}), \ \ \langle
b^\ast(t) b^\ast(t^{\prime}) \rangle = M^\ast \delta(t-t^{\prime}).
\end{eqnarray}
Here $\langle X \rangle$ is a standard notation used to denote the quantum
expectation of a system operator $X$ \cite{EM98,AFP09} (i.e., $\langle X \rangle=\mathbb{E}[X]$), $N \geq 0$ is the average number of bosons, while $M$ describes
the amount of squeezing in the field state. We have the identity $|M|^2 \leq
N(1+N)$. For a thermal state, $M=0$ and 
\begin{equation}
N = \frac{1}{ e^{\beta (E-\mu)}-1} ,
\end{equation}
where $\beta =\frac{1}{k_B T}$ is the inverse temperature, $E$ is the
energy, and $\mu$ is the chemical potential.

In this paper we will assume $N=M=0$, which corresponds to the case of a
boson field in the vacuum (ground) state. The vacuum boson field is a
natural quantum extension of white noise, and may be described using the
quantum It\={o} calculus.  In this calculus, the integrated field processes $%
B(t)= \int_0^t b(s) ds$ (annihilation), $B^\ast(t)=\int_0^t b^\ast(s)ds$
(creation) and $\Lambda(t) = \int_0^t b^\ast(s)b(s) ds$ (counting) are used.
The non-zero It\={o} products for the vacuum boson field are 
\begin{equation}
d\Lambda(t) d\Lambda(t)=d\Lambda(t), \ \ d\Lambda(t) dB^\ast(t)=dB^\ast(t),
\ \ dB(t)d\Lambda_0(t)=dB(t), \ \ dB(t) dB^\ast(t)=dt.
\end{equation}

 We now specify the  fermion channels $A_0$ and $A_1$  in Figure \ref{fig:fermion-model1}.
We assume the followings singular expectations for a fermion field $A$, defined
on an antisymmetric Fock space $\mathfrak{F}_{antisym}$: 
\begin{eqnarray}
\langle a^\ast(t) a(t^{\prime}) \rangle &=& N \delta(t-t^{\prime}),\ \
\langle a(t) a^\ast(t^{\prime}) \rangle = (1- N) \delta(t-t^{\prime}), \\
\langle a(t) a(t^{\prime}) \rangle &=& M \delta(t-t^{\prime}),\ \ \langle
a^\ast (t) a^\ast(t^{\prime}) \rangle = M^\ast \delta(t-t^{\prime}).
\label{eqn:Fermi_cov}
\end{eqnarray}
In general we have $0 \leq N \leq 1$ along with the identity $|M|^2 \leq
N(1-N)$. For a thermal state we have $M=0$, and 
\begin{equation}
N = \frac{1}{ e^{\beta (E-\mu)}+1}.
\end{equation}

In what follows we take the zero temperature limit $T \to 0$. For fermion
channel 1 we assume the energy is such that $E < \mu$ and so in the zero
temperature limit this channel is fully occupied,  $N=1$, and    the It\={o} rule 
\begin{equation}
dA_1^\ast (t) dA_1(t)=dt
\end{equation}
applies for the corresponding integrated processes $A_1(t)= \int_0^t a_1(s) ds$, $%
A_1^\ast(t)=\int_0^t a_1^\ast(s)ds$.   For fermion
channel 0 we fix $E > \mu$, in which case $N=0$, describing a reservoir
which is unoccupied. The number process $\Lambda_0(t) = \int_0^t a_0^\ast(s)
a_0(s) ds$ is well defined for fermion channel 0 (but not for channel 1),
and the It\={o} table is 
\begin{equation}
d\Lambda_0(t) d\Lambda_0(t)=d\Lambda_0(t), \ \ d\Lambda_0(t)
dA_0^\ast(t)=dA_0^\ast(t), \ \ dA_0(t)d\Lambda_0(t)=dA_0(t), \ \ dA_0(t)
dA_0^\ast(t)=dt.
\end{equation}
The fermion channels are defined on distinct antisymmetric Fock spaces $%
\mathfrak{F}_{antisym}^{(1)}$, $\mathfrak{F}_{antisym}^{(0)}$.

\subsection{System Coupled to Boson and Fermion Fields}
\label{sec:model-system}

The system $S$ illustrated in Figure \ref{fig:fermion-model1} is defined on the Hilbert space $\mathfrak{H}_S$, and so the complete system coupled to the boson and fermion fields  is defined on the tensor product Hilbert space
\begin{equation}
\mathfrak{H}= \mathfrak{H}_S \otimes \mathfrak{F}_{sym} \otimes  \mathfrak{F}_{antisym}^{(1)} \otimes\mathfrak{F}_{antisym}^{(0)} .
\end{equation}
Due to the presence of fermion field channels, it is necessary to introduce a parity structure on the collection of operators on this tensor product space, as explained in Appendix  \ref{sec:mixed-structure}. We therefore have a
   parity operator $\tau$ on $\mathfrak{H}$  
such that for all operators $X$ and $Y$ on $\mathfrak{H}$ we have $\tau(XY)=\tau(X) \tau(Y)$ and $\tau(X^\ast)=\tau(X)^\ast$. Operators $X$ such that $\tau(X)=X$ are called {\em even}, while those for which $\tau(X)=-X$ are called {\em odd}. Fermion annihilation and creation operators are odd, while the fermion number operator is even. All boson operators are even. A system operator, i.e. an operator $X$ acting nontrivially on $\mathfrak{H}_S$ only, that is even will commute with all field operators, while an odd system operator will anticommute with odd fermion field operators.  All boson field operators commute with all system operators and all fermion field operators.

The Schr\"{o}dinger equation for the complete system is
\begin{eqnarray}
dU(t) &=& (  (S-I)d\Lambda(t)  + dB^\ast(t) L  - L^\ast S dB(t) -\frac{1}{2} L^\ast L dt
\nonumber
\\
&& + dA_1^\ast(t) L_1  - L_1^\ast dA_1(t) -\frac{1}{2} L_1 L_1^\ast dt
  \nonumber
\\
&&  +   (S_0-I)d\Lambda_0(t)  + dA_0^\ast(t) L_0  - L_0^\ast S_0 dA_0(t) -\frac{1}{2} L_0^\ast L_0 dt 
\nonumber
\\
&&
-iH dt
) U(t),  
\label{eq:qsde-U}
\end{eqnarray}
with initial condition $U(0)=I$. I 
The operators $S$, $L$, $H$, $S_0$, $L_1$ and $L_0$ are system operators, where
\begin{itemize}
\item 
$S$, $L$, $H$, $S_0$ are even (and thus also their adjoints), and 
\item
$L_1$ and $L_0$ are odd (and thus also their adjoints).
\end{itemize}
The operator $H$ is called the Hamiltonian, and it describes the behavior of the system in the absence of field coupling. The operators $S$, $L$, $S_0$, $L_1$ and $L_0$ describe how the field channels couple to the system ($S$ and $S_0$ are required to be unitary).
Note that  terms involving the creation and annihilation operators in (\ref{eq:qsde-U})  ensure a total  energy conserving exchange of energy between the system and the field channels; for example, an electron may transfer from the field to a quantum dot, and vice versa.
Consequences of the specified parity of the above operators and the fact that $U(0)=I$ is even is that $U(t)$ is even and hence commutes with all the It\={o} differentials, and, by the quantum It\={o} rule, is a unitary process (we have $dA_0^\ast L_0 = -L_0 dA_0^*$, $dA_1^*(t)L_1=-L_1dA_1^*(t)$, and  $dB^\ast L = LdB^*$, see Appendix \ref{sec:app-stoch-quantum}, equations 
(\ref{eq:inc-fermion-ac}) and (\ref{eq:inc-boson-c})).

\subsubsection{Heisenberg Picture Dynamics}
A system operator $X$ at time $t$ is given in the Heisenberg picture by $%
X(t)=j_{t}( X) =U( t) ^{\ast }  X U( t) $ and it follows from the
quantum It\={o} calculus and the commutation and anticommutation relations arising from the chosen parity 
that
\begin{eqnarray}
dj_{t}( X) &=&j_{t}(S^{\ast }XS-X) d\Lambda ( t) + dB(t) ^{\ast }    j_{t}(S^{\ast }[X,L]) 
+j_{t}([L^{\ast },X]S) dB( t) +j_{t}(\mathcal{L}(X)) dt 
 \notag \\
&& 
+  dA_1(t) ^{\ast }  j_{t} (\tau(X)L_1-L_1X)  
+j_{t}(L_1^{\ast}\tau(X)-XL_1^\ast) dA_1( t) +j_{t}(\mathcal{L}^\tau_1(X)) dt 
\nonumber \\
&& 
+ j_{t}(S_0^{\ast }XS_0-X) d\Lambda_0 ( t) +dA_0(t) ^{\ast }   j_{t}(S_0^{\ast }(\tau(X)L_0-L_0X)) 
+j_{t}((L_0^{\ast }\tau(X)-XL_0^\ast)S_0) dA_0( t) 
\nonumber \\
&& 
+j_{t}(\mathcal{L}_0^\tau(X)) dt  - i j_t([X,H]) dt ,
 \label{eq:X-dyn} 
\end{eqnarray}
where
\begin{eqnarray}
\mathcal{L}(X) & = & L^\ast X L - \frac{1}{2}XL^\ast L  - \frac{1}{2} L^\ast LX,
\\ 
\mathcal{L}^\tau_1(X) & = & L_1 \tau(X) L_1^\ast - \frac{1}{2}XL_1 L_1^\ast  - \frac{1}{2} L_1 L_1^\ast  X,
\\
\mathcal{L}^\tau_0(X) & = & L^{\ast }_0 \tau(X) L_0 - \frac{1}{2}XL_0^\ast L_0  - \frac{1}{2} L_0^\ast L_0 X ,
\end{eqnarray}
and in the case of even operators we shall just write $\mathcal{L}_i(X)= L^{\ast }_i X L_i - \frac{1}{2}XL_i^\ast L_i  - \frac{1}{2} L_i^\ast L_i X$, ($i=0,1$).

The boson and fermion output fields are defined by
\begin{eqnarray}
B_{out}(t) &=& U^\ast(t) B(t) U(t), \ \  \Lambda_{out}(t) = U^\ast(t) \Lambda(t) U(t), 
\\
A_{1, out}(t) &=& U^\ast(t) A_1(t) U(t),  
\\
A_{0, out}(t) &=& U^\ast(t) A_0(t) U(t), \ \ \Lambda_{0,out}(t) = U^\ast(t) \Lambda_0(t) U(t)
\end{eqnarray}
and satisfy the corresponding quantum stochastic differential equations (QSDEs)
\begin{eqnarray}
dB_{out}(t) &=& j_t(L) dt + j_t(S) dB(t),
\\
d\Lambda_{out}(t) &=& j_t(L^\ast L) dt + dB^\ast(t) j_t(S^*L) + j_t(L^\ast S) dB(t) + d\Lambda(t),
\\
dA_{1,out}(t) &=& j_t(L_1^\ast) dt + dA_1(t),
\\
dA_{0,out}(t) &=& j_t(L_0) dt + j_t(S_0) dA_0(t),
\\
d\Lambda_{0,out}(t) &=& j_t(L_0^\ast L_0) dt + dA_0^\ast(t) j_t(S_0^\ast L_0) + j_t(L_0^\ast S_0) dA_0(t) +   d\Lambda_0 (t).
\end{eqnarray}

\subsubsection{The State}
We define a state $\mathbb{E} [\cdot ]$ on the von Neumann algebra of observables to be an expectation, that is, a linear positive normalized map from the observables to the complex numbers;  positive meaning that $\mathbb{E} [X^*X] \geq 0$ for any observable $X$ and normalized meaning $\mathbb{E} [I] =1$, where $I$ is the identity operator. For technical reasons we require the state to be continuous in the normal topology, see for instance \cite{Meyer}. We shall assume that the state is a product state with respect to the system-environment decomposition: $\mathbb{E} [ X \otimes F ] \equiv \langle X \rangle_S \, \langle F \rangle_E$, for system observable $X$ and environment observable $F$. In particular we take $\langle \, \cdot \, \rangle_E$ to be the mean zero gaussian state with covariance (\ref{eqn:Fermi_cov}) and the choice of $N=1$ (the Fermi vacuum).

We say that the state is {\em even} if we have
\begin{equation}
\mathbb{E} \circ \tau = \mathbb{E},
\end{equation}
where $\tau$ is the parity operator that was introduced in Section \ref{sec:model-system}. Specifically, this forces all odd observables to have mean zero. In quantum theory, the observable quantities must be self-adjoint operators, however, it is not necessarily true that all self-adjoint operators are observable as there may exists so-called superselection sectors. In the present case, only the even self-adjoint operators are observables. We need to ignore states which lead to unphysical correlations between component systems, this is referred to a a superselection principle in the quantum physics literature \cite{Giulini}. We need therefore to restrict our interest to even states only. More specifically, we shall assume that the factor states $\langle \, \cdot \, \rangle_S$ and $\langle \, \cdot \, \rangle_E$ are separately even on the system and environment observables respectively.

The expected values of system operators $X$ evolve in time as follows. Define
 \begin{equation}
\mu_t(X) = \mathbb{E}[j_t(X) ].
\end{equation}
Then by taking expectations of (\ref{eq:X-dyn}) we find that for even observables $X$
 \begin{eqnarray}
\dot{\mu}_t(X) = \mu_t( \mathcal{L}(X)  + \mathcal{L}_1(X)  + \mathcal{L}_0(X) ),
\end{eqnarray}
which is called the {\em master equation}, and corresponds to the Kolmogorov equation (\ref{eq:X-mean-c}).
This may be expressed in Schr\"{o}dinger form using the density operator $\rho(t)$ defined by $\mu_t(X) = \mathrm{tr}[ \rho(t) X]$, which exists by our assumption of normal continuity of the state. The density operator is then an \emph{even} positive trace-class operator, normalized so that $\mathrm{tr}[ \rho(t) ]=1$, satisfying the equation
\begin{equation}
\dot{\rho}(t) = \mathcal{L}^\ast(\rho(t)) + \mathcal{L}_1^\ast(\rho(t))  + \mathcal{L}_0^\ast(\rho(t)),
\label{eq:master}
\end{equation}
where
\begin{eqnarray}
\mathcal{L}^\ast(\rho) & = & L \rho L^\ast - \frac{1}{2} L^\ast L\rho   - \frac{1}{2} \rho L^\ast L ,
\\ 
\mathcal{L}_1^\ast(\rho) & = &  L_1^\ast \rho L_1 - \frac{1}{2} L_1 L_1^\ast \rho  - \frac{1}{2} \rho L_1 L_1^\ast   ,
\\
\mathcal{L}_0^\ast(\rho) & = & L_0  \rho L_0^\ast - \frac{1}{2} L_0^\ast L_0\rho  - \frac{1}{2} \rho L_0^\ast L_0   .
\end{eqnarray}
 
\section{Fermion Filter}
\label{sec:filter}

In this section we suppose that electrons in fermion channel 0, after interaction with the system, can be continuously counted; that is, the observables $\Lambda_{0, out}(s)$, $0 \leq s \leq t$, are  measured. The problem is, given an even state $\mathbb{E}$ as outlined above, to determine estimates $\hat X(t)$ of system operators $X$ given the measurement record. This is a filtering problem involving a signal derived from a fermion field. As mentioned above only the even operators may be observable, and in fact the expectation and conditional expectation of all odd operators must vanish identically.

Mathematically, we wish to determine equations for the quantum conditional expectations
\begin{equation}
\hat X(t) = \pi_t(X) = \mathbb{E}[ j_t(X) \, \vert \, \mathscr{Y}_t ]  .
\label{eq:pi-def}
\end{equation}
Here, $X$ is a system operator, $\mathscr{Y}_t$ is the algebra generated by the operators $\Lambda_{0, out}(s)$, $0 \leq s \leq t$, a commutative von Neumann algebra, and $\pi_t$ is the conditional state. In quantum mechanics, conditional expectations are not always well defined due to the general lack of commutativity. However, the conditional expectation (\ref{eq:pi-def}) is well defined because $j_t(X)$ commutes with all operators in the algebra $\mathscr{Y}_t$. This is called the non-demolition property, and is a consequence of the system-field model, where fermion field channel 0  serves as a probe, see \cite{BHJ07}.  The quantum conditional expectation (\ref{eq:pi-def}) is characterized by the requirement that
\begin{equation}
\mathbb{E}[ j_t(X) Z ] = \mathbb{E}[ \pi_t(X) Z ] \ \mathrm{for \ all} \ Z \in \mathscr{Y}_t.
\label{eq:pi-char}
\end{equation}

\begin{theorem}  \label{thm:filter}
The quantum filter for the conditional expectation (\ref{eq:pi-def}) is given by $\pi_t (X)=0$ for odd observables, while for even observables satisfies the equation
\begin{eqnarray}
d \pi_t(X) &=& \pi_t( -i[X,H] + \mathcal{L}(X) +  \mathcal{L}_1(X)  + \mathcal{L}_0(X)  ) dt
\nonumber \\
&&
+ \biggl\{  \frac{\pi_t(L^\ast_0 X L_0)}{\pi_t(L_0^\ast L_0)} - \pi_t(X) \biggr\} dW(t)
\label{eq:fermion-filter}
\end{eqnarray}
where $W(t)$ is a $\mathscr{Y}_t$ martingale (innovations process)  given by
\begin{equation}
dW(t) = dY(t) - \pi_t(L_0^\ast L_0) dt, \ \ W(0)=0.
\end{equation}

\end{theorem}

\begin{proof}
We  derive the filtering equation using 
the characteristic function method \cite{HSM05}, \cite{VPB93}, \cite{VPB92},
 whereby we postulate that the filter has
the form
\begin{equation}
d \pi_t( X) = \mathcal{F}_t(X) dt + \mathcal{H}_t(X)dY(t) ,
\label{eq:ansatz-1}
\end{equation}
where $\mathcal{F}_t$ and $\mathcal{H}_t$ are to be determined.

Let $f$ be square integrable, and define a process $c_f$ by
\begin{equation}
dc_f(t) = f(t) c_f(t) dY(t), \ \ c_f(0)=1.
\end{equation}
Then $ c_f(t)$ is adapted to $\mathscr{Y}_t$, and the
requirement (\ref{eq:pi-char})  implies that
\begin{equation}
\mathbb{E}[ X(t) c_f(t) ] = \mathbb{E}[ \pi_t(X)   c_f(t)  ]
\label{eq:Xc-1}
\end{equation}
holds for all $f$. We will use this relation to find $\mathcal{F}_t$ and $\mathcal{H}_t$.

The differential of the LHS of (\ref{eq:Xc-1}) is, using the quantum stochastic differential equation (\ref{eq:X-dyn}) and the quantum It\={o} rule,
\begin{eqnarray}
d \mathbb{E}[ X(t) c_f(t) ]  &=&  \mathbb{E}[ (d j_t(X)) c_f(t) + j_t(X) dc_f(t)  + dj_t(X) dc_f (t) ]
\\
&=& \mathbb{E}[ c_f(t)  j_t( \mathcal{L}(X) +  \mathcal{L}_1(X)  +  \mathcal{L}_0(X) ) dt + j_t(X) f(t) c_f(t) dY(t) 
\nonumber \\
&&
+ f(t) c_f(t) j_t(L_0^\ast \tau(X) L_0 - X L_0^\ast L_0) dt]
  \\
&=&  \mathbb{E}[ c_f(t)  j_t( \mathcal{L}(X) +  \mathcal{L}_1(X)  +  \mathcal{L}_0(X) ) dt  
+ f(t) c_f(t) j_t(L_0^\ast \tau(X) L_0  ) dt]
\end{eqnarray}
Now using the property (\ref{eq:Xc-1}) we find that
\begin{eqnarray}
\frac{d}{dt}  \mathbb{E}[ X(t) c_f(t) ] 
&=&  \mathbb{E}[ c_f(t)  \pi_t( \mathcal{L}(X) +  \mathcal{L}_1(X)  +  \mathcal{L}_0(X) )    
+ f(t) c_f(t) \pi_t(L_0^\ast \tau(X) L_0  )  ]
\label{eq:ansatz-2}
\end{eqnarray}

Next, the differential of the RHS of (\ref{eq:Xc-1}), using the ansatz (\ref{eq:ansatz-1}) and the quantum It\={o} rule, is
\begin{eqnarray}
\frac{d}{dt} \mathbb{E}[ \pi_t(X)   c_f(t)  ] &=& \mathbb{E}[ c_f(t) ( \mathcal{F}_t(X) + \mathcal{G}_t(X) \pi_t(L_0^\ast L_0) ) 
\nonumber \\
&&
+ f(t) c_f(t) (\pi_t(X) \pi_t( L_0^\ast L_0) + \mathcal{G}_t(X) \pi_t( L_0^\ast L_0)
) ]
\label{eq:ansatz-3}
\end{eqnarray}
Equating coefficients of $c_f$ and $f c_f$ in equations (\ref{eq:ansatz-2}) and (\ref{eq:ansatz-3}) gives the equations
\begin{eqnarray}
\pi_t( \mathcal{L}(X) +  \mathcal{L}_1(X)  +  \mathcal{L}_0(X) )     &=&  \mathcal{F}_t(X) + \mathcal{G}_t(X) \pi_t(L_0^\ast L_0) 
\\
\pi_t(L_0^\ast \tau(X) L_0  ) &=&
\pi_t(X) \pi_t( L_0^\ast L_0) + \mathcal{G}_t(X) \pi_t( L_0^\ast L_0)
\end{eqnarray}
from which the filter coefficients are readily determined, and we deduce the full filter equations
\begin{equation}
d \pi_t(X) = \pi_t( -i[X,H] + \mathcal{L}(X) +  \mathcal{L}_1(X)  + \mathcal{L}_0(X)  ) dt
+ \biggl\{  \frac{\pi_t(L^\ast_0 \tau(X) L_0)}{\pi_t(L_0^\ast L_0)} - \pi_t(X) \biggr\} dW(t).
\end{equation}
Taking $X$ to be odd and even in turn yields the result.
\end{proof}

\begin{corollary} \label{cor:filter}
Let $\rho_0$ be the initial even density matrix for the system, then in the Schr\"{o}dinger picture we may define the conditional density operator $\hat \rho(t)$ by 
 $\pi_t(X)=  \mathrm{tr}[\hat \rho(t)X]$, and obtain the filtering equation
\begin{eqnarray}
d\hat \rho(t)&=(\mathcal{L}^\ast(\hat \rho(t)) + \mathcal{L}_1^\ast(\hat \rho(t))  + \mathcal{L}_0^\ast(\hat \rho(t)) dt+
(\displaystyle{ \frac{L_0\hat \rho(t) L_0^*}{{\rm tr}(L_0 \hat \rho(t)L_0^*)}   }  -\hat \rho(t)) dW(t),
\label{eqn:SME}
\end{eqnarray}
with $\hat{\rho}(0) =\rho_0$.
\end{corollary}

\section{Examples}
\label{sec:examples}

In this section we consider several examples drawn from the literature. These examples are special cases of the model described above (Figure \ref{fig:fermion-model1}).

\subsection{Quantum Dot}
\label{sec:eq-dot}

We consider a quantum dot arrangement discussed in \cite[sec. 3]{GM00}, as shown in Figure \ref{fig:quantum-dot-1}. 
The left ohmic contact $L$ is assumed to be a perfect emitter, which we describe by a fermion field channel $A_1(t)$ (for which we have 
$dA_1^\ast (t)dA_1(t)= dt$), while the right ohmic contact $R$ is assumed to be a perfect absorber, given by a fermion field channel $A_0(t)$ ($dA_0 (t)dA_0^\ast(t)= dt$).
Current flows through the quantum dot by tunneling. The quantum dot is described by  annihilation and creation operators $c$ and $c^\ast$, respectively, satisfying $\{ c, c^\ast \}=1$ and $c^2=0$. The dot is coupled to the two fermi channels via the  operators $L_1= i \sqrt{\gamma_L} \, c$, $L_0 = i \sqrt{\gamma_R}\, c$, and $S_0=I$. Here, $\gamma_L$ and $\gamma_R$ are the tunneling rates across the left and right barriers, respectively. We take $H=0$, and   there is no boson field channel in this example. The parity is defined such that $c$ and $c^\ast$ are odd.

The dynamics of the quantum dot are described by a differential equation for $c$; from (\ref{eq:X-dyn}) we have 
\begin{equation}
dc(t) = -\frac{\gamma_L+\gamma_R}{2} c(t) dt + i \sqrt{\gamma_L}\, dA_1(t) + i \sqrt{\gamma_R}\, dA_0(t) .
\label{eq:c-dyn-1}
\end{equation}
Equation (\ref{eq:c-dyn-1}) is a linear quantum stochastic differential equation, with solution
\begin{eqnarray}
c(t) = e^{-\frac{\gamma_L+\gamma_R}{2}t} c + \int_0^t e^{-\frac{\gamma_L+\gamma_R}{2}(t-s)} (i \sqrt{\gamma_L}\, dA_1(s) + i \sqrt{\gamma_R}\, dA_0(s)).
\end{eqnarray}
Note, however, that this system is not Gaussian. The influence of the two fermion fields on the dot can be seen in these equations through the stochastic integral terms.

The output fields are given by
\begin{eqnarray}
dA_{1,out}(t) &=& -i  \sqrt{\gamma_L}\, c^\ast (t) dt + dA_1(t)
\\
dA_{0,out}(t) &=& i  \sqrt{\gamma_R}\, c (t) dt + dA_0(t)
\end{eqnarray}
and
\begin{eqnarray}
d\Lambda_{0,out}(t) = \gamma_R n(t) dt + i \sqrt{\gamma_R}\, dA_0^\ast (t) c(t) -i \sqrt{\gamma_R}\, c^\ast(t) dA_0(t) + d\Lambda_0(t),
\end{eqnarray}
where $n(t) = c^\ast(t) c(t)$ is the number operator (an even operator) for the quantum dot.  The output field components exhibit contributions for the dot and the input fields.
Using (\ref{eq:X-dyn}), the number operator $n(t)$ satisfies the equation
\begin{equation}
dn(t) = \gamma_L(1-n(t)) dt - \gamma_R n(t) dt + i \sqrt{\gamma_L}\, (c^\ast(t) dA_1(t) - dA_1^\ast(t) c(t) ) +
 i \sqrt{\gamma_R}\, (c^\ast(t) dA_0(t) - dA_0^\ast(t) c(t) ) 
\end{equation}

We turn now to the expected behavior of the quantum dot system, \cite{SM99}, \cite{GM00}. The differential equation for the unconditional density operator $\rho(t)$ is, from (\ref{eq:master}):
\begin{eqnarray}
\dot{\rho}(t) = \gamma_L (  c^\ast \rho(t) c  - \frac{1}{2} cc^\ast \rho(t)  - \frac{1}{2} \rho (t)  cc^\ast)
+ \gamma_R(  c\rho (t)  c^\ast - \frac{1}{2} c^\ast c \rho(t)  - \frac{1}{2} \rho (t)  c^\ast c ) .
\end{eqnarray}
The expected number of fermions in the dot is defined by $\bar{n}(t) = \mathbb{E}[ n(t) ]$, 
and satisfies the differential equation
\begin{equation}
\frac{ d \bar{n}(t)}{dt} = \gamma_L (1-\bar{n}(t)) - \gamma_R \bar{n}(t).
\end{equation}
In  steady state, the average  number of fermions in the dot is $\bar{n}_{ss} = \gamma_L/(\gamma_L + \gamma_R)$, reflecting an equilibrium balance of electron flow through the source and sink channels.

Now suppose that the current in the right contact is continuously monitored; this corresponds to the output field observable $\Lambda_{0,out}(t)$. We may condition on this information to obtain an estimate of the quantum dot occupation number,
$
\hat n(t) = \mathbb{E} [ n(t)  \vert  \mathscr{Y}_t ] = \pi_t(n).
$
Using (\ref{eq:fermion-filter}), the stochastic differential equation for $\hat n(t)$ is
\begin{eqnarray}
d \hat n(t) = \gamma_L (1- \hat n(t) ) dt - \gamma_R \hat n(t) dt - \hat n(t) ( dY(t) - \gamma_R \hat n(t)dt) .
\end{eqnarray}
It is worth comparing the form of this filtering equation to the classical Kalman filter (\ref{eq:kalman-1}), and the quantum filter for an atom monitored by a boson field, (\ref{eq:boson-filter-n}).  While the details of the dynamics differs in these cases, the filters share the same structure, with an additive correction term involving an innovations process. The \lq{gain}\rq \  in this correction term is not deterministic, unlike the case of the Kalman filter (\ref{eq:kalman-1}) for the conditional mean $\hat \xi(t)$.

\subsection{Photodetection}
\label{sec:detect}

A {\em photodetector} is a sensing device that produces an electronic current flow in response to light incident upon it. At the quantum level, a discrete output results from the arrival of a photon.
A simple model for a photodetector involving both boson and fermion fields is described in \cite[sec. 8.5]{GZ00}.   In this section we use this model to describe the detection of photons scattered from an atom, and then we show how to use the information in the electron flow to estimate atomic variables using a fermion filter.

A schematic representation of the detection of the photons  emitted by an atom is shown in Figure \ref{fig:photo-detector-1}. This figure illustrates that the output boson channel for the atom is fed into the input boson channel of the detector. The atom is modeled as a two level system on the Hilbert space $\mathbf{C}^2$ (Section \ref{sec:model-qm}) coupled to a boson field (Section \ref{sec:model-fields}). The atom has lowering and raising operators $\sigma_-$ and $\sigma_+=\sigma_-^\ast$, respectively, and our interest is in the atomic observable $n = \sigma_+\sigma_-$ 
counting the quanta in the atom (0 or 1).
The detector is modeled as a three-level system with Hilbert space $\mathbf{C}^3$ coupled to boson and fermion field channels. The analogs of the lowering and raising operators are the operators
\begin{equation}
\sigma_{jk} = \vert j \rangle \langle k \vert, \ \ j,k=1,2,3,
\end{equation}
where $\vert 1 \rangle$, $\vert 2 \rangle$ and $\vert 3 \rangle$ denote a basis for $\mathbf{C}^3$; thus  $\vert j \rangle = \sigma_{jk} \vert k \rangle$, as indicated by the arrows in Figure \ref{fig:photo-detector-1}.

The connection of the atom to the detector via the boson channel is an instance of a cascade or series connection, \cite{CWG93}, \cite{HJC93}, \cite{GJ09}. If the time delay between the components is small relative to the other timescales involved, then a single markovian model may be used for the combined atom-detector system. In this model, the operators $\sigma_{32}$ and $\sigma_{13}$ are odd, while 
$\sigma_-$, $\sigma_+$, $\sigma_{12}$, $\sigma_{11}$, $\sigma_{22}$ and $\sigma_{33}$ are even. The Hamiltonians for both subsystems is taken to be zero.
We take $H=\frac{i}{2} \sqrt{\kappa\gamma}\,(\sigma_{12}\sigma_+ - \sigma_{21} \sigma_-)$ (arising from the series connection), $S=S_0=I$ and set the coupling operators to be
$L=\sqrt{\kappa}\, \sigma_- + \sqrt{\gamma} \, \sigma_{12}$, $L_0 = \sqrt{\gamma_{0}}\, \sigma_{32}$, $L_1 = \sqrt{\gamma_{1}}\, \sigma_{31}$.
The quantum stochastic  equations of motion for  the atom are
   \begin{eqnarray}
d \sigma_-(t) &=& -\frac{\kappa}{2} \sigma_-(t) +   \sqrt{\kappa}\, (2 n(t)-I) dB(t)
\\
d n(t) &=& -\kappa n(t) dt - \sqrt{\kappa}\, ( dB^\ast(t) \sigma_-(t) + \sigma_+(t) dB(t))
\end{eqnarray}
and for the detector 
\begin{eqnarray}
d \sigma_{11}(t) &=& (\gamma \sigma_{22}(t) + \gamma_1 \sigma_{33}(t) )dt
+\sqrt{\kappa\gamma}\, ( \sigma_+(t) \sigma_{12}(t) + \sigma_{21}(t) \sigma_-(t)
)dt 
\nonumber \\ 
&&
+ \sqrt{\gamma}\, (d B^\ast(t) \sigma_{12}(t) + \sigma_{21}(t) d B(t))
-\sqrt{\gamma_1}\, (dA_1^\ast(t) \sigma_{31}(t) - \sigma_{13}(t) dA_1(t))
\\
d \sigma_{22}(t) &=& - (\gamma  + \gamma_0) \sigma_{22}(t) dt 
 - \sqrt{\kappa\gamma}\, ( 
\sigma_+(t) \sigma_{12}(t) + \sigma_{21}(t) \sigma_-(t)
)dt
\nonumber \\ 
&&
- \sqrt{\gamma}\, (d B^\ast(t) \sigma_{12}(t) - \sigma_{21}(t) d  B(t))
-\sqrt{\gamma_0}\, (dA_0^\ast(t) \sigma_{32}(t) - \sigma_{23}(t) dA_0(t))
\\
d \sigma_{33}(t) &=& \gamma_0 \sigma_{22}(t)dt  - \gamma_1 \sigma_{33}(t) dt
\nonumber \\ 
&&
+\sqrt{\gamma_0}\, ( dA_0^\ast(t) \sigma_{32}(t) + \sigma_{23}(t) dA_0(t))
+\sqrt{\gamma_1}\, ( dA_1^\ast(t) \sigma_{31}(t) + \sigma_{13}(t) dA_1(t))
\\
d \sigma_{12}(t) &=&  - \frac{1}{2} (\gamma + \gamma_0) \sigma_{12}(t) dt 
+  \frac{1}{2}\sqrt{\kappa\gamma}\, (\sigma_{22}(t)-\sigma_{11}(t)) \sigma_-(t) dt
\nonumber \\
&&
+\frac{1}{2}\sqrt{\gamma}\, (\sigma_{22}(t)-\sigma_{11}(t))d  B(t)
-\sqrt{\gamma_0}\, \sigma_{13}(t) dA_0(t)
+\sqrt{\gamma_1}\, \sigma_{32}(t) dA_1^*(t)
\\
d \sigma_{32}(t) &=& - \frac{1}{2} (\gamma + \gamma_0+\gamma_1) \sigma_{32}(t) dt 
 - \sqrt{\kappa\gamma}\, \sigma_{31}(t) \sigma_-(t) dt
\nonumber \\
&&
-\sqrt{\gamma}\, \sigma_{31}(t) d B(t)
-\sqrt{\gamma_0}\, ( \sigma_{22}(t) + \sigma_{33}(t) ) dA_0(t)
-\sqrt{\gamma_1}\, \sigma_{12}(t) dA_1(t)
\\
d \sigma_{13}(t) &=&  -\frac{\gamma_1}{2} \sigma_{13}(t) dt
 + \sqrt{\kappa\gamma}\, \sigma_{23}(t) \sigma_-(t) dt
 \nonumber \\
 &&
+\sqrt{\gamma}\, \sigma_{23}(t) d  B(t)
-\sqrt{\gamma_0}\, \sigma_{12}(t) dA_0^\ast(t)
-\sqrt{\gamma_1}\, ( \sigma_{11}(t)+  \sigma_{33}(t)  ) dA_1^\ast (t)
\end{eqnarray}
The number operator for the output of fermion channel $0$ evolves according to
\begin{equation}
d\Lambda_{0,out}(t) = \gamma_0 \sigma_{22}(t) dt + \sqrt{\gamma_0}\, ( dA_0^\ast(t) \sigma_{12}(t) + \sigma_{21}(t) dA_0(t))
+ d\Lambda_0(t)
\end{equation}

The mean value $\bar n (t) = \mathbb{E}[ n(t) ]$ evolves according to
\begin{equation}
\dot{\bar n}(t) = -\kappa \bar n(t) ,
\end{equation}
and so $\bar n(t) \to 0$ as $t \to \infty$. Thus initial quanta in the atom are lost to the fields.
 
To calculate the conditional mean $\hat n(t) = \mathbb{E}[ n(t) \vert \mathscr{Y}_t ]$, we will make use of the notation
 \begin{equation}
\sigma_{jk\alpha\beta} = \sigma_{jk} \sigma_{\alpha} \sigma_{\beta}
\end{equation}
where $j,k=1,2,3$ and $\alpha,\beta=-,+$ and the fact that atomic operators commute with detector operators.
The conditional mean $\hat n(t)$ is given by the system of equations
 \begin{eqnarray}
d \hat n(t) &=& -\kappa \hat n(t) dt + \biggl(\frac{\hat \sigma_{22+-}(t)}{\hat \sigma_{22}(t)} - \hat n(t)\biggr) dW(t)
\\
d \hat \sigma_{22}(t) &=& -(\gamma + \gamma_0) \hat \sigma_{22}(t) dt - \sqrt{\kappa\gamma}\, (\hat \sigma_{12+}(t) + \hat \sigma_{12+}^\ast(t) ) dt -\hat \sigma_{22}(t) dW(t)
\\
d \hat \sigma_{12+}(t) &=& -\frac{1}{2}(\kappa+ \gamma + \gamma_0) \hat \sigma_{12+}(t) dt - \sqrt{\kappa\gamma}\, \hat \sigma_{11+-}(t) dt -\hat \sigma_{12+}(t) dW(t)
\\
d \hat \sigma_{11+-}(t) &=&   -\kappa  \hat \sigma_{11+-}(t) dt + \gamma  \hat \sigma_{22+-}(t) dt + \gamma_1  \hat \sigma_{33+-}(t) dt 
-\hat \sigma_{11+-}(t) dW(t) \ \ 
\\
d \hat \sigma_{22+-}(t) &=&   -(\kappa+ \gamma + \gamma_0) \hat \sigma_{22+-}(t) dt -\hat \sigma_{22+-}(t) dW(t)
\\
d \hat \sigma_{33+-}(t) &=&  -(\kappa + \gamma_1) \hat \sigma_{33+-}(t) dt + \gamma_0 \hat \sigma_{22+-}(t) + 
 \biggl(\frac{\hat \sigma_{22+-}(t)}{\hat \sigma_{22}(t)} - \hat \sigma_{33+-}(t)\biggr) dW(t) 
\end{eqnarray}
where the innovations process is given by
\begin{equation}
dW(t) = dY(t) - \gamma_0 \hat \sigma_{22}(t) dt.
\end{equation}
Note that the filter involves estimates of variables associated with the detector.

\section{Conclusion}
In this paper, using the boson and fermion quantum stochastic calculus we have derived the quantum filtering equations for a class of open quantum systems (which may be fermionic) that are coupled to both bosonic and fermionic fields, for the case where the measurement (driving the filter) is that of counting of electrons in a fermionic field. For illustration, the filtering equations were calculated for two examples of estimating the number of electrons in a quantum dot coupled to an electron source and sink, and that of counting the photons emitted by a two level atom via  a photodetector which is modelled as a fermionic three level system. For both of these examples we find that the resulting set of filtering equations is closed (i.e., the  quantum filter is completely determined by a finite number of coupled matrix stochastic differential equations). 
 
 \appendix
 
 \section{Stochastic Calculus}
 \label{sec:app-stoch}

 \subsection{Classical}
\label{sec:app-stoch-classical}

Let   $w(t)$ be a Brownian motion (Wiener process). This means that $w(t)$ is an independent increment process and $w(t)-w(s)$ is a Gaussian  random variable with zero mean and variance $t-s$.
Suppose that
 $\alpha(t)$ is an adapted  stochastic process, that is,   
 $\alpha(t)$ is independent of $w(r)$ for all $r > t$; in particular, the It\={o} increment $dw(t) = w(t+dt)-w(t)$ is independent of $\alpha(t)$.
 The It\={o} stochastic integral of $\alpha$ with respect to $w$ 
 is  defined as a limit involving  forward  increments:
 \begin{equation}
I(t) = \int_0^t \alpha(s) dw(s) \approx \sum_j \alpha(s_j) (w(s_{j+1}) -w(s_j)) ,
\end{equation}
where, $s_0=0 < s_1 < s_2 < \cdots \leq t$. 
Since the Wiener process has zero mean, so does the stochastic integral:
$
\mathbf{E}[ I(t) ] =0.
$

Now suppose we have two stochastic integrals
 \begin{equation}
I(t) = \int_0^t \alpha(s) dw(s)  \ \text{and} \ J(t) = \int_0^t \beta(s) dw(s) ,
\end{equation}
where $\alpha$ and $\beta$ are adapted. Then the {\em It\={o} product rule} says that
\begin{equation}
I(t)J(t) = \int_0^t (J(s) \alpha(s) + I(s) \beta(s)) dw(s) + \int_0^t \alpha(s) \beta(s) ds,
\end{equation}
which is the sum of a stochastic integral  and conventional (say Lebesgue) integral, called the It\={o} correction term (the last term on the right hand side).   The It\={o} correction term is  not present in the usual product rule.

Stochastic integrals and the product rule are often expressed in differential form. Indeed, we may write
\begin{equation}
dI = \alpha dw \ \text{and} \ dJ = \beta dw,
\end{equation}
so that
\begin{eqnarray}
d(IJ) &=&  (dI)J + I(dJ) + (dI)(dJ)
\\
&=& J \alpha dw + I\beta dw + \alpha\beta dt ,
\end{eqnarray}
where we see that the correction term arises from the It\={o} rule
 \begin{equation}
dw(t) dw(t) = dt
\end{equation}
which is very useful in calculations.

 \subsection{Quantum}
\label{sec:app-stoch-quantum}

Let   $B(t)$, $B^\ast(t)$ be    boson annihilation and creation operators, as discussed in Section \ref{sec:model-fields}. For simplicity here we ignore the number (counting) field  operators and assume   $M=0$ and $N=0$ (vacuum field state).

 Let $\alpha_1(t)$ and $\alpha_2(t)$ be operator-valued adapted  processes, i.e.  independent of future field operators.  Quantum It\={o}  stochastic integrals are defined in terms of forward increments:
\begin{eqnarray}
I(t) &= & \int_0^t \alpha_1(s) dB(s) + \int_0^t \alpha_2(s) dB^\ast(s) 
\notag \\
& \approx & \sum_j \alpha_1(s_j) ( B(s_{j+1}) -B(s_j)) +   \sum_j \alpha_2(s_j) ( B^\ast(s_{j+1}) -B^\ast(s_j))  .
\end{eqnarray}
The expected value of the above stochastic integral is zero in the vacuum state $\vert \phi \rangle$:
$
\mathbb{E}_{\phi} [ I(t) ] =0.
$

  The quantum It\={o} rule is expressed in terms of four products
\begin{eqnarray}
dB(t) dB(t) = 0, \ \ dB(t) dB^\ast(t) = dt, \
dB^\ast(t) dB(t) = 0, \ \ dB^\ast(t) dB^\ast(t) = 0.
\end{eqnarray}
This It\={o} table is valid for vacuum and coherent field states. It\={o} tables for squeezed and thermal field states have more non-zero terms, as mentioned in Section \ref{sec:model-fields}.
An important property is that for an adapted process $\alpha(t)$, we have
\begin{equation}
[ \alpha(t), dB(t) ] =0=[ \alpha(t), dB^\ast(t) ] .
\label{eq:inc-boson-c}
\end{equation}

Now suppose we have two quantum stochastic integrals,  $I(t)$ defined above and
\begin{equation}
J(t) = \int_0^t \beta_1(s) dB(s) + \int_0^t \beta_2(s) dB^\ast(s)  .
\end{equation}
The product rule is
\begin{eqnarray}
d(IJ) &=&  (dI)J + I(dJ) + (dI)(dJ)
\\
&=&  \alpha_1 J dB + \alpha_2 J dB^\ast + I \beta_1 dB + I \beta_2 dB^\ast + \alpha_1 \beta_2 dt
\\
&=& ( \alpha_1 J   + I \beta_1 )dB + (  \alpha_2 J     + I \beta_2) dB^\ast + \alpha_1 \beta_2 dt;
\end{eqnarray}
that is,
\begin{eqnarray}
I(t) J(t) &=&  \int_0^t ( \alpha_1(s) J(s)   + I(s) \beta_1(s) )dB(s) + \int_0^t (  \alpha_2(s) J(s)     + I(s) \beta_2(s)) dB^\ast(s) 
\notag \\
&& + \int_0^t \alpha_1(s) \beta_2(s) ds.
\end{eqnarray}
The last term in this expression is the It\={o} correction terms, and arises from the non-zero It\={o} product $dB(t) dB^\ast(t)=dt$.
Note that the order is important in these expressions, since the expressions involve quantities that need not commute.

Quantum stochastic integrals with respect to fermion fields may also be defined, \cite{AH84}, \cite{CWG04}. However, some extra effort is required to keep track of antisymmetric tensor products and parity, matters that are explained in Appendix \ref{sec:mixed-structure}.  Here we mention that for an {\em odd} adapted process $\beta(t)$ and a fermion field $A(t)$, $A^\ast(t)$ we have the anti-commutation relations
 \begin{equation}
\{   \beta(t), dA(t) \} =0= \{ \beta(t), dA^\ast(t) \} .
\label{eq:inc-fermion-ac}
\end{equation}

 \section{Mixed Boson-Fermion Systems and Generalized Antisymmetric Tensor Product of Operators}
\label{sec:mixed-structure}

The laws of quantum physics dictate that  fermionic operators from independent systems must {\em anticommute} with one another. That is, if $X_j$ is any fermionic operator on a system $j$ with Hilbert space $\mathfrak{H}_j$
for $j=1,2$ then we must have $\{X_1,X_2\}=X_1X_2 +X_2X_1=0$. Realizing this property is not trivial because typically when one has a composite system with Hilbert space $\mathfrak{H}_1  \otimes \mathfrak{H_2}$ (here $\otimes$, depending on the context, denotes the tensor product between two Hilbert spaces or the tensor product between two Hilbert space operators) the operator $X_1$, originally defined only on $\mathfrak{H}_1$, would be identified with its ampliation $X_1 \tensor I$ on the composite Hilbert space, whereas $X_2$ would be identified with $I \otimes X_2$. This is  problematic for fermionic systems  since then we would have that $\{X_1 \otimes I,I  \otimes X_2\} = 2 X_1 \otimes X_2$, which is not necessarily $0$ for arbitrary fermionic operators $X_1$ and $X_2$. To resolve this issue and get these operators to anticommute correctly, the Hilbert spaces of fermionic systems are endowed with  some  additional structure and the usual tensor product between operators must be replaced with the antisymmetric tensor product between operators. We briefly explain this below, for a more detailed account see, for instance,   \cite{AH84}.

Let $\mathfrak{H}$ be the Hilbert space of a fermionic system. Then it is required that there exists an orthogonal decomposition  $\mathfrak{H}=\mathfrak{H}_+ \oplus \mathfrak{H}_-$  where the subspace $\mathfrak{H}_+$ is called the {\em even subspace}, while $\mathfrak{H}_-$ is called the 
{\em odd subspace}. An operator $X$ on $\mathfrak{H}$ is said to be even if it leaves the even and odd subspaces invariant $X \mathfrak{H}_{\pm} \subset  \mathfrak{H}_{\pm}$ and odd if it maps vectors of the even subspace into the odd subspace and vice-versa $X \mathfrak{H}_{\pm} \subset  \mathfrak{H}_{\mp}$. 
In particular the 
fermionic operators (or degrees of freedoms) are odd operators.

Let $P_+$ and $P_-$ denote the orthogonal projection onto the odd and even subspace, respectively. Define the linear {\em parity operator} $\theta = P_+-P_-$ and the linear {\em parity superoperator} $\tau(X) = \theta X \theta$ for any bounded operator $X$ on $\mathfrak{H}$ . 
It is easily checked from the definition that $\theta$ and $\tau$ have the following properties:
\begin{enumerate}
\item $\theta^{*} =\theta$.

\item $\theta^{*}\theta = \theta^2 =I$.

\item $\tau(XY) = \tau(X)\tau(Y)$
\item $\tau(X^*) = \tau(X)^*$ 
\item If $X$ is even then $\tau(X)=X$, while if $X$ is odd then $\tau(X)=-X$.
\end{enumerate}
Note that properties $3$ to $5$ imply that $\tau$ is an automorphism on $\mathcal{B}(\mathfrak{H})$ (the space of all bounded operators on $\mathfrak{H}$). For an operator $X$ with definite parity 
(i.e., it is either even or odd) we define the binary-value parity  functional 
$\delta(X)$ as  $\delta(X)=0$ if $X$ is even and $\delta(X)=1$ if $X$ is odd. Thus for any operator with a definite parity we have that $\tau(X) = (-1)^{\delta(X)}X$. In general an operator $X$ has a decomposition into odd part and even parts  as $X=X_+ + X_-$, where $X_+$ and $X_-$ are even and odd operators, respectively, given by:
\begin{align*}
X_+ &=  P_+ X P_+ + P_- X P_- \\
X_- &=  P_+  X P_- + P_-  X P_+ .
\end{align*}
Thus 
$$
\tau(X) = \tau(X_+) + \tau(X_-)= (-1)^{\delta(X_+)} X_+ + (-1)^{\delta(X_-)}X_-= X_+ - X_-.
$$

We can now define the composition of independent 
fermionic systems $(\mathfrak{H}_{1}, \theta_{1})$ and 
$(\mathfrak{H}_{2},\theta_{2})$. The composite system has Hilbert space  $\mathfrak{H}_{12}=\mathfrak{H}_1 \otimes \mathfrak{H}_2$ and parity operator  $\theta_{12} = \theta_1 \otimes \theta_2$ so that its even and odd subspaces are  $H_{12+} = H_{1+} \otimes H_{2+} \oplus  H_{1-} \otimes H_{2-}$ and $H_{12-} = H_{1-} \otimes H_{2+} \oplus H_{1+} \otimes H_{2-}$. 
The key question is how to ampliate the action of operators $X_{j}$ acting on the individual Hilbert spaces $\mathfrak{H}_{j}$ to the composite Hilbert space, so as to satisfy the desired fermionic anti-commutation relations. For this we define the antisymmetric tensor product $\hat{\otimes}$ between $X_1$ and $X_2$ as 
$$
X_1 \hat{\otimes} X_2= X_1 \otimes X_{2+} + X_{1} \theta_1 \otimes X_{2-} .
$$ 
The ampliation of an operator $X_{1}$ of the first system is $X_{1}\hat{\otimes} I= X_{1}\otimes I$, and the ampliation of an operator $X_{2}$ of the second system is $I \hat{\otimes} X_2= I \otimes X_{2+} + \theta_1 \otimes X_{2-}$, such that 
$X_1 \hat{\otimes} X_2 =(X_{1}\hat{\otimes} I )(I \hat{\otimes} X_2)$ . In particular if both $X_1$ and $X_2$ are odd then $\{ X_1 \hat{\otimes} I,I \hat{\otimes} X_2\}=0$ (the ampliations anticommute), while if at least one of them is even then $[X_1 \hat{\otimes} I,I \hat{\otimes} X_2]=0$ (the ampliations commute). 
Note that the ampliation preserves  the original parities of the system operators, e.g. $ \tau_{12} ( I \hat{\otimes} X_{2}) =\theta_{12} (I \hat{\otimes} X_{2})\theta_{12}=  I\hat{\otimes}  \tau_{2}(X_{2}) $.   

%
%


Now, if $(\mathfrak{H}_3, \theta_{3})$ is a third fermionic system we can 
show that $(X_1 \hat{\otimes} X_2) \hat{\otimes} X_3 = X_1 \hat{\otimes} ( X_2 \hat{\otimes} X_3)$ so that the antisymmetric tensor product is associative and $X_1 \hat{\otimes} X_2 \hat{\otimes} X_3$ is unambigously defined. By repeating the above procedure we can define the composition of an arbitrary number of systems $(\mathfrak{H}_j, \theta_{j}) $, $j=1,2,\dots,n$   such that 
$\hat{\otimes}_{j=1}^n X_j$ is well defined for arbitrary operators $X_j$ on $\mathfrak{H_j}$. For simplicity we  identify each $X_j$ with its ampliation to $\otimes_{j=1}^n \mathcal{H}_j$ with respect to the antisymmetric tensor product $\hat{\otimes}$, so that if $j \neq k$ then $\{X_j,X_k\}=0$ whenever both $X_j$ and $X_k$ are odd, while $[X_j,X_k]=0$ whenever one of them is even.

%
%
      
In the setting of this paper, we also consider composite systems consisting of fermionic and non-fermionic (bosonic) sub-systems. The composition can be described in the framework of the anti-symmetric tensor product by endowing the non-fermionic systems with a {\it trivial} parity structure. Indeed according to quantum physics,  bosonic operators from one system commute with all operators (even or not)  from the other systems, so they can be interpreted as ``even'' operators. This can be achieved by defining the parity operator of a bosonic system as $\theta=I$ (i..e., take $P_+=I$ and $P_-= 0$)  so that $\tau(X)= X$ for all bosonic operators. 
%
%
%
In this way we can extend the definition of  $\hat{\otimes}$ {\em mutatis mutandis} to mixtures of fermionic  and non-fermionic systems  in the way we had done above for fermionic systems.  For instance, let $\mathfrak{H}_1$ and $\mathfrak{H}_3$ be the Hilbert spaces of two non-fermionic systems and $\mathfrak{H}_2$ is the Hilbert space of a fermionic system. Consider the composite system on $\mathfrak{H}_1 \otimes \mathfrak{H}_2$ and let $X_j$ be an operator on $\mathfrak{H}_j$. Then we have that  $X_1 \hat{\otimes} X_2 = X_1 \otimes X_{2+} + X_1 \otimes X_{2-}=X_1 \otimes X_2$. Similary,  for the composite system on $\mathfrak{H}_2 \otimes \mathfrak{H}_3$ we have $X_2 \hat{\otimes} X_3= X_2 \otimes X_3$. 
That is, the generalized antisymmetric tensor product between two operators from distinct systems of which one is fermionic and the other is not, reduces to the usual tensor product between operators. On the other hand, if the two operators are from two distinct fermionic systems then it  reduces to the antisymmetric tensor product for operators of fermionic systems.  This is  exactly how it should be: non-fermionic operators from one system commute with all operators  from the other systems, and the fermionic ones anti-commute with fermionic operators from the other systems. By inspection it is easy to see that the generalized antisymmetric tensor product defined in this way has the associative property, hence it is unambiguously defined for systems that are the composite of any finite number of fermionic and non-fermionic systems.

{\em Acknowledgement:}
JG would like to acknowledge kind hospitality the Australian National University during a research visit in February 2010 where part of this paper was written, along with the support of the UK Engineering and Physical Sciences Research Council, under project EP/G039275/1, and the Australian Research Council. MRJ and HIN would likewise acknowledge the support of the Australian Research Council (including APD fellowship support for HIN), and 
of EPSRC project EP/H016708/1, for research visits to Abersystwyth University. MG would like to acknowledge kind hospitality the Australian National University during a research visit in July 2007, and  the support of an EPSRC Fellowship EP/E052290/1.
 
\bibliographystyle{plain}


\end{document}